%% file: m.tex
\documentclass[1+]{elsarticle}

\usepackage{amsmath}
\usepackage{listings}
\usepackage{tikz}
\usetikzlibrary{automata,positioning, arrows}
\usetikzlibrary{patterns}
\usetikzlibrary{decorations,decorations.pathmorphing,decorations.markings}

\usepackage{logic9}

\usepackage{todonotes}

\newcommand{\lleq}{\lessdot}
\newcommand{\seq}{\mathit{seq}}
\newcommand{\ftalu}{\sqsubseteq_{LU}}
\newcommand{\abslu}{\mathit{abs_{LU}}}
\newcommand{\aaa}{\mathfrak{a}}
\newcommand{\alu}{\mathfrak{a}_{\fleq \scriptscriptstyle{LU} }}

\newcommand{\lu}{\fleq_{\scriptscriptstyle LU}}

\renewcommand{\int}[1]{\lfloor #1 \rfloor}

\newcommand{\xra}{\xrightarrow}

\newcommand{\ceil}[1]{\lceil #1 \rceil}
\newcommand{\Closure}{\operatorname{\mathit{Closure}}}
\newcommand{\Extra}{\mathit{Extra}}

\newcommand{\abstr}{\mathfrak{a}}
\newcommand{\abs}{\abstr}
\newcommand{\tto}{\Rightarrow}

\newcommand{\tuple}[1]{\ensuremath{(#1)}}
\newcommand{\Rpos}[0]{\ensuremath{\mathbb{R}_{\geq 0}}}

\renewcommand{\Nat}[0]{\ensuremath{\mathbb{N}}}
\renewcommand{\CC}[0]{\ensuremath{\Phi}} 
\newcommand{\val}[0]{\ensuremath{v}} 
\newcommand{\vali}[0]{\ensuremath{\mathbf{0}}} 
\newcommand{\intpart}[1]{\ensuremath{\lfloor #1\rfloor}}
\newcommand{\fracpart}[1]{\ensuremath{\{ #1\}}}

\newcommand{\reset}[1]{\ensuremath{[#1]}}

\newcommand{\elup}[1]{\Extra^+_{LU}(#1)}
\newcommand{\Repeat}[2]{\noindent{$\blacktriangleright$\textsc{\textbf{#1~\ref{#2}.\
      }}}}

\newcommand{\tasim}{\preceq_{t.a.}}
\newcommand{\abssim}{\abs_{\tasim}}
\newcommand{\nbd}[1]{\operatorname{nbd}(#1)}

\newcommand{\ggeq}{\gtrdot}

\newcommand{\rrlu}[1]{\langle #1 \rangle^{\scriptscriptstyle LU}}
\newcommand{\gv}{G_v^{\scriptscriptstyle M}}
\newcommand{\gvlu}{G_v^{\scriptscriptstyle LU}}

\newcommand{\reg}[1]{[#1]^{\scriptscriptstyle M}}
\newcommand{\lureg}[1]{\langle #1 \rangle^{\scriptscriptstyle LU}}

\newcommand{\erlu}[1]{\rrlu{#1}}

\lstdefinelanguage{algo}{%
  morekeywords={function,algorithm,push,pop,top,for,all,and,or,if,then,else,repeat,until,while,do,report,return,such,that,each,add,call,exit,let}
}

\begin{document}

\begin{frontmatter}

  \title{Better Abstractions for Timed Automata\tnoteref{averts}}

  \author[labri]{Fr\'ed\'eric Herbreteau} \ead{fh@labri.fr}

  \author[cmi]{B. Srivathsan\corref{cor1}} \ead{sri@cmi.ac.in}

  \author[labri]{Igor Walukiewicz} \ead{igw@labri.fr}

  \address[labri]{Universit\'e de Bordeaux, Bordeaux INP, CNRS, LaBRI,
    UMR 5800, Talence, France}

  \address[cmi]{Chennai Mathematical Institute, Chennai, India}

  \tnotetext[averts]{This work has been supported by project AVeRTS -
    CEFIPRA - Indo-French Program in ICST - DST/CNRS
    ref. 218093. Author B. Srivathsan is partially funded by a grant
    from Infosys Foundation.}

  \cortext[cor1]{Corresponding author}
  
\begin{abstract}
  We study the reachability problem for timed automata. A standard
  solution to this problem involves computing a search tree whose
  nodes are abstractions of zones. These abstractions preserve
  underlying simulation relations on the state space of the automaton.
  For both effectiveness and efficiency reasons, they are
  parameterized by the maximal lower and upper bounds ($LU$-bounds)
  occurring in the guards of the automaton.
  
  One such abstraction is the $\alu$ abstraction defined by Behrmann
  et al. Since this abstraction can potentially yield non-convex sets,
  it has not been used in implementations. Firstly, we prove that
  $\alu$ abstraction is the coarsest abstraction with respect to
  $LU$-bounds that is sound and complete for reachability. Secondly,
  we provide an efficient technique to use the $\alu$ abstraction to
  solve the reachability problem.
\end{abstract}


\end{frontmatter}

\input{intro} \input{prelims} \input{biggest-lu} \input{alu-abslu}
\input{inclusion} \input{conclusions}

\section*{References} 
\bibliographystyle{plain} \bibliography{m}

\end{document}

%% file: intro.tex
\section{Introduction}

Timed automata are finite automata extended with clocks whose values
can be compared with constants and set to zero. The clocks measure
delays between different steps of execution of the automaton. The
reachability problem for timed automata asks if there exists a path
from its initial state to a given target state. This problem cannot be
solved by a simple state exploration since clocks are real-valued
variables.  The standard solution to this problem involves computing
the zone graph of the automaton which in principle could be
infinite. In order to make it finite, zones are approximated using an
abstraction operator. Till recently it has been generally assumed that
for reasons of efficiency an abstraction of a zone should always be a
zone. Here we avoid this assumption. We first show that $\alu$
abstraction defined by Behrmann et al.~\cite{Behrmann:STTT:2006} is
the coarsest sound and complete abstraction. We then present a method
of constructing the abstracted zone graph using the $\alu$
abstraction.  Even though this abstraction can yield non-convex sets,
we show that our method is at least as efficient as any other
currently known method based on abstractions.

The reachability problem is a basic problem in verification. It is
historically the first problem that has been considered for
timed-automata, and it is still a lively subject of
research~\cite{Behrmann:STTT:2006,Wang:STTT:2004,Morbe:CAV:2011,Herbreteau:FSTTCS:2011}. Apart
from being interesting by itself, the advances on this problem may
give new methods for verification of more complicated models, like
priced timed-automata~\cite{BFLM-cacm11}, or probabilistic timed
automata~\cite{gregersen-jensen,bouyer-hab2009,kwiatkowska2002automatic}.

All approaches to solving the reachability problem for timed automata
should ensure termination. To tackle this, most of them use
abstractions to group together bisimilar valuations of clock
variables, that is, valuations not distinguishable by the automaton as
far as reachability to the final state is concerned.  The first
solution has been based on regions, which are certain equivalence
classes of clock valuations~\cite{Alur:TCS:1994}. Their definition is
parameterized by a threshold up to which the clock values should be
considered.  A great improvement in efficiency has been obtained by
adopting zones instead of regions. These are sets of valuations
defined by conjunctions of differences between pairs of clocks. They
can be efficiently implemented using difference bound matrices
(DBMs)~\cite{DBLP:conf/ifip/BerthomieuM83,Dill:AVMFSS:1989}.  A
challenge with zone based approach
is that they are not totally compatible with regions, and moreover a
forward exploration algorithm can produce infinitely many zones.  The
union of regions intersecting a zone is a natural candidate for a
finitary abstraction. Indeed this abstraction would make the forward
exploration algorithm terminate. However such an union of regions is
not necessarily a zone and so it is not clear how to represent it. For
this reason a number of abstraction operators have been proposed that
give an approximation of the union of regions intersecting a zone. A
coarser approximation would make the abstracted zone graph smaller. So
potentially it would give a more efficient algorithm.

An important observation made in~\cite{Behrmann:STTT:2006} is that if
reachability is concerned then we can consider simulation instead of
bisimulation.  Indeed, it is safe to add configurations that can be
simulated by those that we have already reached. Simulation relations
in question depend on the given automaton, and it is \EXPTIME-hard to
calculate the biggest one~\cite{larsch00}. A pragmatic approach is to
abstract some part of the structure of the automaton and define a
simulation based on this information. The most relevant information is
the set of bounds with which clocks are compared to in the guards of
the automaton. Since lower and upper bounds are considered separately,
they are called $LU$-bounds.  In~\cite{Behrmann:STTT:2006} the authors
define an abstraction based on simulation with respect to $LU$-bounds;
it is denoted $\alu$. Theoretically $\alu$ is very attractive: it has
clear semantics and, as we show here, it is always a union of regions.
The problem is that $\alu$ abstraction of a zone is seldom a convex
set, so one cannot represent the result as a zone.

In this paper we give another very good reason to consider $\alu$
abstraction. We show that it is actually the coarsest abstraction that
is sound and complete with respect to reachability for all automata
with the same $LU$-bounds. In other words, it means that in order to
get coarser (that is better) abstractions one would need to look at
some other structural properties of automata than just
$LU$-bounds. Our main technical result is an effective algorithm for
dealing with $\alu$ abstraction. It allows to manipulate this
abstraction as efficiently as purely zone based ones.  We propose a
forward exploration algorithm working with zones that constructs the
$\alu$ abstraction of the transition graph of the automaton. This
algorithm uses standard operations on zones, plus a new test of
inclusion of a zone in the $\alu$ abstraction of another zone. The
test is quadratic in the number of clocks and not more complex than
that for just testing an inclusion between two zones. Since $\alu$
abstraction is the coarsest sound and complete abstraction, it can
potentially give smallest abstract systems.

\subsection{Related Work}

Forward analysis is the main approach for the reachability testing of
real-time systems. The use of zone-based abstractions for termination
has been introduced in~\cite{Daws:TACAS:1998}. In recent years,
coarser abstractions have been introduced to improve efficiency of the
analysis~\cite{Behrmann:STTT:2006}. An approximation method based on
LU-bounds, called $\Extra^+_{LU}$, is used in the current
implementation of
UPPAAL~\cite{Behrmann:QEST:2006}. In~\cite{Herbreteau:FSTTCS:2011} it
has been shown that it is possible to efficiently use the region
closure of $\Extra^+_{LU}$, denoted $\Closure^+_{LU}$. This has been
the first efficient use of a non-convex abstraction. In comparison,
$\alu$ approximation has a well-motivated semantics, it is also region
closed, coarser than $\Closure^+_{LU}$ and the resulting inclusion test
is even simpler than that of $\Closure^+_{LU}$. A comparison of these
abstractions is depicted in Fig.~\ref{fig:abst}. The $\alu$ inclusion
test (restricted to the case when $L = U$) has recently been adapted to weighted timed automata and priced
zones~\cite{DBLP:conf/cav/BouyerCM16}.

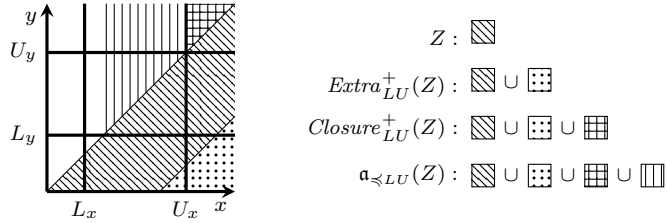
\begin{figure}[!t]
  \centering \input{fig-abst}
  \caption{A comparison of abstraction operators for zones.}
  \label{fig:abst}
\end{figure}

Let us mention that abstractions are not needed in backward
exploration of timed systems. Nevertheless, any feasible backward
analysis approach needs to simplify constraints. For
example~\cite{Morbe:CAV:2011} does not use approximations and relies
on an SMT solver instead. Clearly this approach is very difficult to
compare with the forward analysis approach we study here.

Another related approach to verification of timed automata is to build
a quotient graph of the semantic graph of the automaton with respect
to some bisimulation
relation~\cite{Tripakis:FMSD:2001,alur1992minimization,Yannakakis:FMSD:1997}.
For reachability properties, this approach is not a priori competitive
with respect to using the simulation-based abstraction $\alu$. It is
more adapted to checking branching time properties.

\subsection{Organization of the Paper}

The paper is organized as follows.

\begin{description}
\item \textit{Section 2} presents the preliminary definitions. It
  introduces the notion of sound and complete abstractions
  parameterized by LU-bounds. It also explains how these abstractions
  could be used to solve the reachability problem.

\item \textit{Section 3} proposes $\abslu$ abstraction, and proves
  that it is the coarsest sound and complete abstraction for all
  automata with given $LU$-bounds.

\item \textit{Section 4} shows that the $\alu$ abstraction
actually coincides with this coarsest abstraction
$\abslu$.

\item \textit{Section 5} presents an efficient inclusion test for
  $\alu$ abstraction, which allows for its use in implementations. 
\end{description}

\medskip 

A preliminary version of this paper appeared in the conference on
Logic in Computer Science~\cite{LICS}. This version includes all  
missing proofs and gives a more elaborate discussion about the inclusion
test (Section~\ref{sec:an-oox2-algorithm}).



%% file: fig-abst.tex
{
  \footnotesize
  \newlength{\N}
  \setlength{\N}{5cm}
  \newlength{\Lx}
  \setlength{\Lx}{1cm}
  \newlength{\Ux}
  \setlength{\Ux}{3.7cm}
  \newlength{\Ly}
  \setlength{\Ly}{1.5cm}
  \newlength{\Uy}
  \setlength{\Uy}{3.7cm}
  \newlength{\Dxy}
  \setlength{\Dxy}{3cm}
  \begin{tikzpicture}[scale=0.5]
    \coordinate (bl) at (0,0);
    \coordinate (tr) at (\N,\N);
    \begin{scope}[->,very thick,>=stealth]
      \draw (bl) -- ([yshift=\N]bl);
      \draw (bl) -- ([xshift=\N]bl);
    \end{scope}
    \node[below left] (x) at ([xshift=\N,yshift=-0.5mm]bl) {$x$};
    \node[below left] (y) at ([xshift=-0.5mm,yshift=\N]bl) {$y$};
    \begin{scope}[very thick]
      \draw ([xshift=\Lx]bl) -- ([xshift=\Lx,yshift=\N]bl); 
      \draw ([xshift=\Ux]bl) -- ([xshift=\Ux,yshift=\N]bl); 
      \draw ([yshift=\Ly]bl) -- ([xshift=\N,yshift=\Ly]bl); 
      \draw ([yshift=\Uy]bl) -- ([xshift=\N,yshift=\Uy]bl); 
    \end{scope}
    \node[below] (Lx) at ([xshift=\Lx,yshift=-0.5mm]bl) {$L_x$};
    \node[below] (Ux) at ([xshift=\Ux,yshift=-0.5mm]bl) {$U_x$};
    \node[left] (Ly) at ([xshift=-0.5mm,yshift=\Ly]bl) {$L_y$};
    \node[left] (Uy) at ([xshift=-0.5mm,yshift=\Uy]bl) {$U_y$};
    \path[fill,pattern=north west lines]
    (bl) --
    ([xshift=\N,yshift=\N]bl) --
    ([xshift=\N,yshift=\N-\Dxy]bl) --
    ([xshift=\Dxy]bl);
    \draw (bl) -- ([xshift=\N,yshift=\N]bl); 
    \draw ([xshift=\Dxy]bl) -- ([xshift=\N,yshift=\N-\Dxy]bl); 
    \path[fill,pattern=dots]
    ([xshift=\Dxy]bl) --
    ([xshift=\N,yshift=\N-\Dxy]bl) --
    ([xshift=\N]bl);
    \path[fill,pattern=grid]
    ([xshift=\Ux,yshift=\Uy]bl) --
    ([xshift=\Ux,yshift=\N]bl) --
    ([xshift=\N,yshift=\N]bl);
    \path[fill,pattern=vertical lines]
    ([xshift=\Ly,yshift=\Ly]bl) --
    ([xshift=\Ly,yshift=\N]bl) --
    ([xshift=\Ux,yshift=\N]bl) --
    ([xshift=\Ux,yshift=\Ux]bl);
    \begin{scope}[xshift=7.4cm,yshift=5.6cm,scale=1.5]
      \node[left] (Z) at (2.5,-1) {$Z:$};
      \draw[fill,pattern=north west lines] (2.6,-1.1) rectangle (3,-0.7);
      \node[left] (ExtraLUZ) at (2.5,-1.85) {$\elup{Z}:$};
      \draw[fill,pattern=north west lines] (2.6,-1.95) rectangle
      (3,-1.55);
      \node at (3.3,-1.8) {$\cup$};
      \draw[fill,pattern=dots] (3.6,-1.95) rectangle (4,-1.55);
      \node[left] (ClosureLUZ) at (2.5,-2.6) {$\Closure_{LU}^+(Z):$};
      \draw[fill,pattern=north west lines] (2.6,-2.8) rectangle
      (3,-2.4);
      \node at (3.3,-2.6) {$\cup$};
      \draw[fill,pattern=dots] (3.6,-2.8) rectangle (4,-2.4);
      \node at (4.3,-2.6) {$\cup$};
      \draw[fill,pattern=grid] (4.6,-2.8) rectangle (5,-2.4);
      \node[left] (aLUZ) at (2.5,-3.45) {$\alu(Z):$};
      \draw[fill,pattern=north west lines] (2.6,-3.65) rectangle
      (3,-3.25);
      \node at (3.3,-3.45) {$\cup$};
      \draw[fill,pattern=dots] (3.6,-3.65) rectangle (4,-3.25);
      \node at (4.3,-3.45) {$\cup$};
      \draw[fill,pattern=grid] (4.6,-3.65) rectangle (5,-3.25);
      \node at (5.3,-3.45) {$\cup$};
      \draw[fill,pattern=vertical lines] (5.6,-3.65) rectangle
      (6,-3.25);
    \end{scope}
  \end{tikzpicture}
}


%% file: prelims.tex
\section{Preliminaries}
\label{sec:preliminaries}

After recalling some preliminary notions, we introduce a concept of
abstraction as a means to reduce the reachability problem for
timed-systems to the one for finite systems. We then observe that
simulation relation is a convenient way of obtaining abstractions with
good properties. 

\subsection{Timed automata and the reachability problem}
Let $X$ be a set of clocks, i.e., variables that range over $\Rpos$,
the set of non-negative real numbers. A \emph{clock constraint} is a
conjunction of constraints $x\# c$ for $x\in X$,
$\#\in\{<,\leq,=,\geq,>\}$ and $c\in \Nat$, e.g. $(x \le 3 \wedge y >
0)$. Let $\CC(X)$ denote the set of clock constraints over clock
variables $X$.  A \emph{clock valuation} over $X$ is a function
$\val\,:\,X\rightarrow\Rpos$. We denote by $\Rpos^X$ the set of clock
valuations over $X$, and by $\vali$ the valuation that associates $0$ to
every clock in $X$. We write $\val\sat \phi$ when $\val$ satisfies
$\phi\in \CC(X)$, i.e. when every constraint in $\phi$ holds after
replacing every $x$ by $\val(x)$. For $\d\in\Rpos$, let $\val+\d$ be
the valuation that associates $\val(x)+\d$ to every clock $x$. For
$R\subseteq X$, let $\reset{R}\val$ be the valuation that sets $x$ to
$0$ if $x\in R$, and that sets $x$ to $\val(x)$ otherwise.

A \emph{Timed Automaton (TA)} is a tuple $\Aa=\tuple{Q,q_0,X,T,\Acc}$
where $Q$ is a finite set of states, $q_0\in Q$ is the initial state,
$X$ is a finite set of clocks, $\Acc\subseteq Q$ is a set of accepting
states, and $T\,\subseteq\, Q\times \CC(X)\times 2^X \times Q$ is a
finite set of transitions $\tuple{q,g,R,q'}$ where $g$ is a
\emph{guard}, and $R$ is the set of clocks that are \emph{reset} on
the transition.

The semantics of $\Aa$ is a transition system of its configurations.
A \emph{configuration} of $\Aa$ is a pair $(q,v)\in Q\times\Rpos^X$
and $(q_0,\vali)$ is the \emph{initial configuration}. We have two
kinds of transitions:

\smallskip

\par\noindent\textbf{Delay:} $(q,v)\to^\d(q,v+\d)$ for some $\d\in
\Rpos$;

\smallskip

\par\noindent\textbf{Action:} $(q,v)\to^t(q',v')$ for some transition
$t = (q,g,R,q')\in T$ such that $v\sat g$ and $v'=[R]v$.

\smallskip

We will denote by $\Ss_\Aa$ the transition system describing the semantics of a
timed automaton $\Aa$. In this paper we are interested in the \emph{reachability problem}:
does there exist a configuration $(q,v)$ with accepting state
$q \in Acc$ that is reachable from $(q_0,\vali)$ by some finite
sequence of delay and action transitions?

The class of TA we consider is usually known as diagonal-free TA since
clock comparisons like $x-y\leq 1$ are disallowed. Notice that if we
are interested in state reachability, considering timed automata
without state invariants does not entail any loss of generality as the
invariants can be added to the guards. For state reachability, we can
also consider automata without transition labels.

\subsection{Abstractions}

Since the transition system determined by the automaton is infinite,
we usually try to find a finite approximation of it by grouping
valuations together. In consequence we work with configurations
consisting of a state and a set of valuations. For every transition
$t$, we have a transition:
\begin{align*}
  (q,W) \tto^t (q',W') \quad \text{where $W'=\set{v' :\exists v\in W,
      ~\exists \d \in \Rpos \text{ s.t. } \ v\to^t \to^\d v'}$}
\end{align*}
We will write $\tto$ without superscript to denote the union of all
$\tto^t$ relations.
The transition relation defined above considers each valuation $v \in
W$ that can take the transition $t$, obtains the valuation after the
transition, and then collects the time-successors from this obtained
valuation. Therefore the symbolic transition $\tto$ always yields sets
closed under time-successors.  

The initial configuration of the
automaton is $(q_0, \vali)$. Starting from the initial valuation
$\vali$ the set of valuations reachable by a delay at the
initial state are given by $\{ \vali + \d ~|~ \d \in \Rpos\}$. Call
this $W_0$. 

From $(q_0, W_0)$ as the initial node, computing the
symbolic transition relation $\tto$ leads to different nodes $(q, W)$
wherein the sets $W$ are closed under time-successors. Although the
transition relation $\tto$ talks about sets of valuations and not
valuations themselves, it could still be potentially infinite. A
further grouping of valuations is necessary to get finiteness.

An \emph{abstraction operation}~\cite{Behrmann:TACAS:2003} is a
convenient way of expressing a grouping of valuations. It is a
function $\abstr:\Pp(\Rpos^{|X|})\to\Pp(\Rpos^{|X|})$ such that
$W\incl\abstr(W)$ and $\abstr(\abstr(W))=\abstr(W)$. An abstraction
operator defines an abstract semantics:
\begin{displaymath}
  (q,W) \tto_{\abstr }(q',\abstr(W'))
\end{displaymath}
where $\abstr(W)=W$ and $(q,W)\tto (q',W')$.

If $\abstr$ has a finite range then this abstraction is said to be
finite. We write $\tto_\abs^*$ for the transitive closure of
$\tto_\abs$, similarly we write $\tto^*$ and $\to^*$ respectively for
the transitive closure of $\tto$ and $\to$ (where $\to$ denotes the
union of $\to^t$ and $\to^\d$).

Of course we want this abstraction to reflect some properties of the
original system. In order to preserve reachability properties we can
require the following two properties (recall that $\vali \in W_0$):

\smallskip

\par\noindent\textbf{Soundness:} if $(q_0,\abstr(W_0))\tto_{\abstr}^*(q,W)$
then there is a $v\in W$ such that $(q_0,\vali)\to^*(q,v)$.

\smallskip

\par\noindent\textbf{Completeness:} if $(q_0,\vali)\to^*(q,v)$ then
there is a
$W$ such that $v\in W$ and
$(q_0,\abstr(W_0))\tto_{\abstr}^*(q,W)$.

\smallskip

It can be easily verified that if an abstraction satisfies $W\incl
\abstr(W)$ then the abstracted system is complete. However soundness
is more delicate to obtain.

Naturally, it is important to be able to efficiently compute the
abstract transition system. A standard way to do this is to use
zones. A \emph{zone}\label{page:zone def} is a set of valuations
defined by a conjunction of two kinds of constraints: comparison of
differences between two clocks with an integer like $x-y\# c$, or
comparison of a single clock with an integer like $x\# c$, where
$\#\in\{<,\leq,=,\geq,>\}$ and $c\in\Nat$. For instance $(x-y\geq
1)\land(y<2)$ is a zone. Zones can be efficiently represented using
difference bound matrices
(DBMs)~\cite{DBLP:conf/ifip/BerthomieuM83,Dill:AVMFSS:1989}. This
suggests that one should consider abstractions that give zones. As
zones are convex, abstractions that range over zones are called convex
abstractions. This is an important restriction as abstractions based
on regions are usually not convex~\cite{Bouyer:FMSD:2004}.

We propose a way to use non-convex abstractions and zone
representations at the same time. We will only consider sets $W$ of
the form $\abstr(Z)$ and represent them simply by the zone $Z$. This
way we can represent states of an abstract transition system
efficiently: we need just to store a zone. In order for this to work
we need to be able to compute the transition relation $\tto_\abs$ on
this representation. We also need to know when two zone
representations stand for the same node in the abstract system. This
is summarized in the following two requirements:

\smallskip

\par\noindent\textbf{Transition compatibility:} for every transition
$(q,\abstr(Z))\tto_{\abstr}(q',W')$ and the matching transition
$(q,Z)\tto(q',Z')$ we have $W'= \abstr(Z')$.

\smallskip

\par\noindent\textbf{Efficient inclusion test:} for every two zones
$Z,Z'$, the test $Z'\incl \abstr(Z)$ is efficient. Ideally, it has the
same complexity $\mathcal{O}(|X|^2)$ as the inclusion $Z' \incl Z$.

\smallskip

The first condition is quite easy to satisfy: We will show quickly
below that every 
abstraction relation coming from time-abstract
simulation~\cite{Tasiran:CONCUR:1996} is transition compatible. The
second condition is the main topic of the paper.

\begin{definition}[Time-abstract simulation]
  \label{def:time-abstr-simulation}
  A \emph{(state based) time-abstract simulation} between two states
  of a transition system $\Ss_\Aa$ is a relation $(q, v) \tasim (q',v')$
  such that:
  \begin{itemize}
  \item $q = q'$,
  \item if $(q,v) \to^\d (q, v + \d) \to^t (q_1, v_1)$, then there
    exists a $\d' \in \Rpos$ such that $(q,v') \to^{\d'} (q, v' + \d')
    \to^t (q_1, v_1')$ satisfying $(q_1,
    v_1) \tasim (q_1,v_1')$. 
  \end{itemize}
\end{definition}

For two valuations $v,v'$, we say that $v \tasim v'$ if for every
state $q$ of the automaton, we have $(q,v) \tasim (q',v')$.  An
abstraction $\abssim$ based on a simulation $\tasim$ can be defined as
follows:

\begin{definition}[Abstraction based on simulation]
  \label{def:abstr-based-ta-sim}
  Given a set $W$, we define $\abssim(W) = \{ v~|~ \exists v' \in W.\
  v \tasim v'\}$.
\end{definition}

\begin{definition}[Timed-elapsed zone]
  \label{def:time-elapsed-zone}
  A zone $Z$ is said to be \emph{time-elapsed} if it is closed under
  time-successors: that is $Z = \{v + \d ~|~ v \in Z, \d \in \Rpos\}$.
\end{definition}

We had previously noted that all nodes $(q, W)$ reachable through
$\tto^*$ from the initial node $(q_0, W_0)$ are all time elapsed. We
can now show that transition relations coming from abstractions based
on time-abstract simulations satisfy the transition compatibility
condition.

\begin{lemma}\label{lem:transition-compatible}
  Let $\abssim$ be an abstraction based on a time-abstract simulation
  relation and let $Z$ be a time-elapsed zone.  For every transition
  $(q, \abssim(Z)) \tto_{\abssim} (q', W')$ and the matching
  transition $(q,Z) \tto (q',Z')$, we have $W' = \abssim(Z')$.
\end{lemma}

\begin{proof}
  Let $t$ be the transition corresponding to $\tto_{\abssim}$ and
  $\tto$. We first prove that $W' \incl \abssim(Z')$. Let $v' \in
  W'$. We will show that there exists a valuation in $Z'$ that
  simulates $v'$ with respect to $\tasim$.

  By the definition of the abstract symbolic transition $(q,
  \abssim(Z)) \tto_{\abssim} (q', W')$, there is a valuation $v_1 \in
  \abssim(Z)$ and a time elapse $\d_1 \in \Rpos$ such that:
  \begin{align*}
    (q,v_1) \to^t \to^{\d_1} (q', v_1') \text{ and } v' \tasim v_1'
  \end{align*}
  Firstly consider the intermediate configuration obtained after the
  $\to^t$ transition from $(q, v_1)$. Call it $(q', \bar{v}_1')$. We
  know that $\bar{v}_1' \in W'$. This valuation $\bar{v}_1'$ can
  elapse a time $\d_1$ and become $v_1'$. Given that $\tasim$ is a
  time-abstract simulation, this intermediate valuation $\bar{v}_1'$
  can simulate $v'$ too:
  \begin{align}\label{eq:1}
    (q, v_1) \to^t (q', \bar{v}_1') \text{ and } v' \tasim \bar{v}_1'
  \end{align}
  Recall that $v_1 \in \abssim(Z)$. Therefore, there exists a
  valuation $v_2 \in Z$ such that $v_1 \tasim v_2$. As $(q, v_1)$ can
  take the transition $\to^t$, by definition of time-abstract
  simulation, there exists a time elapse $\d_2$ such that $(q, v_2)$
  can take the transition after the time elapse $\d_2$:
  \begin{align}\label{eq:2}
    (q, v_2) \to^{\d_2} \to^t (q', \bar{v}_2') \text{ and } \bar{v}_1'
    \tasim \bar{v}_2'
  \end{align}
  From (\ref{eq:1}) and (\ref{eq:2}), we see that $v' \tasim
  \bar{v}_2'$. Note that as $Z$ is a time-elapsed zone and since $v_2
  \in Z$, we also have $v_2 + \d_2 \in Z$ and this in turn implies
  that $\bar{v}_2' \in Z'$. This shows that $v' \in \abssim(Z')$ and
  hence $W \incl \abssim(Z')$.

  We will now show the converse: $\abssim(Z') \incl W'$.  Let $v \in
  \abssim(Z')$. Then, there exists $v_1 \in Z$ and a $\d_1 \in \Rpos$
  such that $(q,v_1) \to^{t} \to^{\d_1} (q',v_1')$ and $v \tasim
  v_1'$. By the property of an abstraction operator, we will have $v_1
  \in \abssim(Z)$ too. Now, directly by the definition of $(q,
  \abssim(Z)) \tto^{\abssim} (q',W')$, we get that $v \in W'$ and
  hence $\abssim(Z') \incl W'$.

\end{proof}

This paper is essentially about how to satisfy the ``efficient
inclusion test'' condition and get as good abstraction as possible at
the same time. In this context, let us highlight an important remark
that describes when an abstraction is better than another.

\begin{remark}
  \label{rem:coarse-abstractions}
  If $\abstr$ and $\mathfrak{b}$ are two abstractions such that for
  every set of valuations $W$, we have
  $\abstr(W) \incl \mathfrak{b}(W)$, we prefer to use $\mathfrak{b}$
  since the graph induced by $\mathfrak{b}$ is \emph{a priori} smaller
  than the one induced by $\abstr$ (sic \cite{Behrmann:STTT:2006}). In
  such a case, the abstraction $\mathfrak{b}$ is said to be
  \emph{coarser} than abstraction $\abs$.
\end{remark}

\subsection{Bounds as parameters for abstraction}

Remark~\ref{rem:coarse-abstractions} suggests to make use of the coarsest possible
abstraction. For a given automaton it can be
computed if two configurations are in a simulation relation. It should
be noted though that computing the biggest simulation relation is
\EXPTIME-hard~\cite{larsch00}. Since the reachability problem can be
solved in \PSPACE{}, this suggests that it may not be reasonable to
try to solve it using the abstraction based on the biggest
simulation.

We can get simulation relations that are computationally easier if we
consider only a part of the structure of the automaton. The simplest
is to take a simulation based on the maximal constant that appears in
guards. More refined is to take the maximum separately over constants
from lower bound constraints, that is in guards of the form $x>c$ or
$x\geq c$, and those from upper bound constraints, that is in guards
$x<c$ or $x\leq c$.  If one moreover does this for every clock $x$
separately, one gets for each clock two integers $L_x$ and $U_x$. The
abstraction that is currently most used is a refinement of this method
by calculating $L_x$ and $U_x$ for every state of the automaton
separately~\cite{Behrmann:TACAS:2003}. For simplicity of notation we
will not consider this optimization but it can be incorporated with no
real difficulty in everything that follows.  We summarize this
presentation in the following definition.

\begin{definition}[LU-bounds]
  \label{def:LU-bounds}
  The $L$ bound for an automaton $\Aa$ is the function assigning to
  every clock a maximal constant that appears in a lower bound guard
  for $x$ in $\Aa$. Similarly $U$ but for upper bound guards. An
  \emph{LU-guard} is a guard where lower bound guards use only
  constants bounded by $L$ and upper bound guards use only constants
  bounded by $U$. An \emph{LU-automaton} is an automaton using only
  LU-guards.
\end{definition}

In the rest of the paper, we try to find good abstractions
parameterized by LU-bounds that also have an efficient inclusion
test. Section~\ref{sec:bigg-lu-abstr} defines an abstraction $\abslu$
and proves that this is the optimal sound and complete abstraction
that is based on LU-bounds. Section~\ref{sec:abstr-alu-coinc} then
shows that the $\alu$ abstraction is the same as $\abslu$ when zones
closed under time successors are considered. We then give an efficient
test $Z \incl \alu(Z')$ in Section~\ref{sec:an-oox2-algorithm}, which
enables the use of $\alu$ in implementations.


%% file: biggest-lu.tex
\section{The coarsest LU abstraction}
\label{sec:bigg-lu-abstr}

We call an abstraction that is based on LU bounds an \emph{LU
abstraction}. A natural question is to know what is the
coarsest LU-abstraction sound and complete for reachability testing.
Given $L$ and $U$ bounds, we know that the automata under
consideration have guards only of the following form (with $\lleq \in
\{<, \le\}$ and $\ggeq \in \{>, \ge\}$) :
\begin{align*}
  x \ggeq 0,~ x \ggeq 1~, \dots,~ x \ggeq L_x \\
  x \lleq 0, ~x \lleq 1~, \dots,~ x \lleq U_x
\end{align*}
However, we do not know the shape of the automata, in particular, the
order in which the above guards appear in the paths of the automata.

An abstraction $\abs$ is sound if for every possible path using the
above guards that a valuation $v \in \abs(W)$ can
execute, there is a representative $v' \in W$ that can execute the
same path. If this rule is not followed, there is one possible
automaton with guards respecting the given LU-bounds for which this
abstraction is not sound. Hence our question can be reworded as
follows: 

\begin{quote}
  Given $L$ and $U$ bounds, what is the coarsest abstraction that is
  sound and complete for all LU-automata?
\end{quote}

We answer this question in four steps:

\begin{description}
\item \textbf{Step 1.} We define a generic simulation relation
  $\ftalu$ (Definition~\ref{def:lu-simulation}) which is a union over
  all time-abstract simulation relations on LU-automata. Roughly the
  simulation relation says that $v \ftalu v'$ if all paths, using
  LU-guards, executed by $v$ can be executed by $v'$.  We define an
  abstraction $\abslu$ that is based on LU-simulation
  (Definition~\ref{def:abslu}). The definition of LU-simulation is
  difficult to work with as it talks about infinite sequences of
  transitions.

\item \textbf{Step 2.} The next aim is to characterize this
  LU-simulation using a finite sequence of transitions
  (Definition~\ref{def:lu-sequence}).  We want to come up with a
  sequence of LU-guards $\seq(v)$ executed by a valuation $v$ for
  which we can say $v \ftalu v'$ iff $v'$ executes this characteristic
  sequence $\seq(v)$. To achieve this, we go through an intermediate
  definition of what we call LU-regions
  (Definition~\ref{def:lu-region}). We define this sequence in
  Definition~\ref{def:lu-sequence}.

\item \textbf{Step 3.} Steps 1 and 2 have defined the necessary
  notions. We now observe that the following are equivalent
  (Proposition~\ref{prop:rLU-sim}, Corollary~\ref{cor:sequence}):
  \begin{itemize}
  \item $v \ftalu v'$,
    \item $v'$ can execute $\seq(v)$.
  \end{itemize}
 
\item \textbf{Step 4.} The previous step gives a finite
  characterization of the generic LU-simulation $\ftalu$. We use this
  to prove that every sound abstraction should be contained in
  $\abslu$, in other words $\abslu$ is the coarsest abstraction sound
  and complete for all LU-automata (Theorem~\ref{thm:alu-biggest}).
\end{description}

Section~\ref{sec:lu-simulation} handles Step 1;
Section~\ref{sec:lu-regions} defines the LU-regions as mentioned in
Step 2; Sections~\ref{sec:finite-sequ-char}
and~\ref{sec:proof-optimality} handle Steps 3 and 4 respectively.

\subsection{LU-simulation}
\label{sec:lu-simulation}

Using LU-bounds we define a simulation relation on valuations without
referring to any particular automaton; or to put it differently, by
considering all LU-automata at the same time.

\begin{definition}[LU-simulation]\label{def:lu-simulation}
  Let $L$, $U$ be two functions giving an integer bound for every
  clock. The \emph{LU-simulation relation} between valuations is the
  biggest relation $\ftalu$ such that if $v\ftalu v'$ then for every
  LU-guard $g$, and set of clocks $R\incl X$ we have
  \begin{itemize}
  \item if $v\act{g,R} v_1$ for some $v_1$ then $v'\act{g,R} v'_1$ for
    $v'_1$ such that $v_1\ftalu v'_1$.
  \end{itemize}
  where $v\act{g,R} v_1$ means that for some $\d\in\Rpos$ we have
  $v+\d\sat g$ and $v_1=[R](v+\d)$.
\end{definition}

Note that in the above definition, the time elapse $\d'$ required for
$v'$ to satisfy the guard $g$ could be different from the time elapse
$\d$ required for $v$ to satisfy the guard $g$. It is immediate that
$\ftalu$ is the biggest relation that is a time-abstract simulation
for all automata with given LU bounds. We define abstraction operator
$\abslu$ to be the abstraction based on this LU-simulation.

\begin{definition}[Abstraction based on LU-simulation]\label{def:abslu}
  For a zone $Z$ we define: $\abslu(Z)=\set{v ~|~\exists v'\in Z.\
    v\ftalu v'}$.
\end{definition}

The definition of LU-simulation is sometimes difficult to work with
since it talks about infinite sequences of actions.  We will present a
useful characterization implying that actually we need to consider
only very particular sequences of transitions that are of length
bounded by the number of clocks (Corollary~\ref{cor:sequence}).
Essentially, we are interested in the following question: given a
valuation $v$, when does a valuation $v'$ $LU$-simulate it, that is,
when is $v \ftalu v'$. We start with a preparatory definition of what
we call $LU$-regions.

\subsection{LU-regions}
\label{sec:lu-regions}

We introduce the notion of LU-regions. The classical notion of
regions~\cite{Alur:TCS:1994}
depends on the maximum bounds function $M$. Given only the maximum
bounds $M$, we know that there could be guards $x \lleq c$ and $x
\ggeq c$ for $c \in \{ 0, \dots, M_x\}$ in the automaton. Let us call
them the \emph{M-guards}. However,
with the LU-bounds, there is more information and consequently fewer
guards: $x \lleq c$ for $c \in \{ 0, \dots, U_x\}$, and $x \ggeq c$
for $c \in \{ 0, \dots, L_x\}$. Note that for each $x$, we have $M_x =
\max(L_x, U_x)$.

Let us denote the region of a valuation $v$ by $\reg{v}$. A valuation $v'$ belongs to the region $\reg{v}$
if two properties are satisfied:
\begin{description}\label{page:property-regions}
\item \textbf{Invariance by guards:} $v'$ satisfies all $M$-guards that
  $v$ satisfies,
\item \textbf{Invariance by time-elapse:} for every time elapse $\d
  \in \Rpos$, there is a $\d' \in \Rpos$ such that $v' + \d' \in
  \reg{v+\d}$.
\end{description}
We would like to define a notion of LU-regions in the same spirit, now
with the additional information on the guards. For this discussion let
us fix some $L$ and $U$ functions.

\begin{definition}[LU-region]~\label{def:lu-region} For a valuation $v$ we define
  its \emph{LU-region}, denoted $\rrlu{v}$, to be the set of
  valuations $v'$ such that:
  \begin{itemize}
  \item $v'$ satisfies all $LU$-guards that $v$ satisfies.
 
  \item For every pair of clocks $x,y$ with
    $\intpart{v(x)}=\intpart{v'(x)}$, $\intpart{v(y)} =
    \intpart{v'(y)}$, $v(x)\leq U_x$ and $v(y)\leq L_y$ we have:
    \begin{itemize}
    \item if $0< \fracpart{v(x)} < \fracpart{v(y)}$ then
      $\fracpart{v'(x)} < \fracpart{v'(y)}$.
    \item if $0 < \fracpart{v(x)} = \fracpart{v(y)}$ then
      $\fracpart{v'(x)} \leq \fracpart{v'(y)}$.
    \end{itemize}
  \end{itemize}
\end{definition}

\medskip

The first invariance with respect to guards
has been directly incorporated in the first condition of the
definition. The second condition in the definition of LU-regions has
been added in order to obtain the invariance by time-elapse property
mentioned. Note that it is possible to have $v' \in \lureg{v}$ but $v
\notin \lureg{v'}$. The
following lemma will now show that with the two conditions specified
in the definition, one can achieve the invariance with respect to
time-elapse.

\begin{lemma}\label{lem:lu-reg-time}
  Let $v,v'$ be valuations such that $v' \in \rrlu{v}$. For all $\d
  \in \Rpos$, there exists a $\d' \in \Rpos$ such that $v'+\d' \in
  \rrlu{v+\d}$.
\end{lemma}
\begin{proof}
 We are given valuations $v$ and $v'$ such that $v' \in
  \lureg{v}$. Therefore, $v'$ satisfies all the LU-guards that $v$ satisfies, and
  the property given by second condition of
  Definition~\ref{def:lu-region}  is true for the ordering of
  fractional parts. Let us call it the \emph{order property}.
  Additionally, we are given a time elapse $\d \in \Rpos$ from the
  valuation $v$. We need to construct a value $\d' \in \Rpos$ such
  that $v' + \d' \in \lureg{v + \d}$.

  \paragraph{Assume $\d < 1$} Without loss of generality, we can
  assume that $\d < 1$. If $\d \ge 1$, then we can put $\intpart{\d'}
  = \intpart{\d}$ and consider the valuations $v + \intpart{\d}$ and
  $v' + \intpart{\d'}$. As we are not altering the fractional parts in
  these valuations, the order property is true for $v + \intpart{\d}$ and
  $v' + \intpart{\d'}$. It is also easy to see that as $v'$ satisfies
  all LU-guards that $v$ does, the valuation $v' + \intpart{\d'}$
  satisfies all LU-guards that $v + \intpart{\d}$ does. This gives us
  $v' + \intpart{\d'} \in \lureg{v + \intpart{\d}}$ and we need to
  consider the time elapse $\fracpart{\d}$ from $v+
  \intpart{\d}$. Therefore, in the rest of the proof, without loss of
  generality, we can assume that $\d < 1$.

  \paragraph{Assume $\intpart{v(z)} = \intpart{v'(z)}$ for all clocks
    $z$} Suppose for a clock $z$, we have $\intpart{v(z)} <
  \intpart{v'(z)}$. Then, as $v'$ satisfies all guards which $v$ does,
  it should be the case that $U_z < \intpart{v(z)}$. All guards having
  constants in between $\intpart{v(z)}$ and $\intpart{v'(z)}$ would be
  lower bound guards. So, irrespective
  of what we choose for $\d'$, the value $v'(z) + \d'$ will satisfy
  all the LU-guards with respect to $z$ that $v(z) + \d$
  satisfies. Also, $z$ does not concern the order property 
  at all. Similarly, if $\intpart{v'(z)} < \intpart{v(z)}$, as $v'$
  satisfies all LU-guards that $v$ satisfies, it should be the case
  that $L_z < \intpart{v'(z)} < \intpart{v(z)}$. In both the cases, we
  can safely ignore the clock $z$. 

  \paragraph{Assume $\intpart{v(z)} \le \max(L_z, U_z)$ for all clocks
    $z$}  For a clock $z$ such that $\intpart{v(z)}$ is greater than
  both $L_z$ and $U_z$, we know that $v'(z)$ should be greater than
  $L_z$ in order to satisfy the same LU-guards. Hence any amount of
  time elapse would maintain this property and additionally such
  clocks do not concern order property. Hence, we assume without loss of
  generality that all clocks are less than at least one of the bounds.

\paragraph{Constructing $\d'$}
We now have $v$ and $v'$ such that for all clocks $z$, the integral
parts match, that is $\intpart{v(z)} = \intpart{v'(z)}$ and
$\intpart{v(z)} \le \max(L_z, U_z)$. Moreover, the delay is $\d <
1$.

Let $0 \le \l_1 < \l_2 < \cdots < \l_k < 1$ be the fractional parts of
clocks in $v$. Let us denote by $X_i$ the set of clocks $z$ that have
 $\fracpart{v(z)} = \l_i$. Similarly, let $0 \le \l'_1 < \l'_2 <
\cdots < \l'_{k'} < 1$ denote the fractional parts in $v'$ and we
define the set $X'_i$ to be the set of clocks $z$ such that
$\fracpart{v'(z)} = \l'_i$. This is pictorially illustrated below.

\vspace{0.5cm}

\begin{center}
  \begin{tikzpicture}

    \draw (0,0) -- (3,0); \fill (0,0) circle (1pt); \fill (0.5, 0)
    circle (1pt); \fill (1.2, 0) circle (1pt); \fill (2.5, 0) circle
    (1pt); \fill (3,0) circle (1pt);

    \node at (0, -0.2) {\scriptsize $0$}; \node at (0.5, -0.2)
    {\scriptsize $X_1$}; \node at (1.2, -0.2) {\scriptsize $X_2$};
    \node at (2, -0.2) {$\dots$}; \node at (2.5, -0.2) {\scriptsize
      $X_k$}; \node at (3, -0.2) {\scriptsize $1$};

    \node at (1.5, -1) {\scriptsize In $v$};

    \draw (5,0 )-- (8,0);
    \begin{scope}[xshift=5cm]
      \fill (0,0) circle (1pt); \fill (0.5, 0) circle (1pt); \fill
      (1.2, 0) circle (1pt); \fill (2.5, 0) circle (1pt); \fill (3,0)
      circle (1pt);

      \node at (0, -0.2) {\scriptsize $0$}; \node at (0.5, -0.2)
      {\scriptsize $X'_1$}; \node at (1.2, -0.2) {\scriptsize $X'_2$};
      \node at (2, -0.2) {$\dots$}; \node at (2.5, -0.2) {\scriptsize
        $X'_{k'}$}; \node at (3, -0.2) {\scriptsize $1$}; \node at
      (1.5, -1) {\scriptsize In $v'$};

    \end{scope}

  \end{tikzpicture}
\end{center}

\vspace{0.3cm}

After a delay $\d$ from $v$, some of the clocks cross the next
integer, whereas some of them do not. Let us say that clocks in $X_j
\cup \dots \cup X_k$ have crossed the integer. Now the fractional
parts of these clocks would be smaller than those of $X_1 \cup \dots
\cup X_{j-1}$ as shown below:

\begin{center}
  \begin{tikzpicture}[scale = 1.3]

    \draw (0,0) -- (3,0); \fill (0,0) circle (1pt); \fill (0.3, 0)
    circle (1pt); \fill (1.1, 0) circle (1pt); \fill (1.6, 0) circle
    (1pt); \fill (2.5, 0) circle (1pt); \fill (3,0) circle (1pt);

    \node at (0, -0.2) {\scriptsize $0$}; \node at (0.3, -0.2)
    {\scriptsize $X_j$}; \node at (0.7, -0.2) {$\dots$}; \node at
    (1.1, -0.2) {\scriptsize $X_{k}$}; \node at (1.6, -0.2)
    {\scriptsize $X_1$}; \node at (2, -0.2) {$\dots$}; \node at (2.5,
    -0.2) {\scriptsize $X_{j-1}$}; \node at (3, -0.2) {\scriptsize
      $1$};

    \node at (1.5, -0.7) {\scriptsize In $v + \d$};
  \end{tikzpicture}
\end{center}

We need to choose a value $\d'$ so that for all clocks $y \in X_j \cup
\cdots \cup X_k$ such that $(v + \d)(y) \le L_y$, the time elapse $\d'$
takes $v'$ to the next integer. Correspondingly for all the clocks $x
\in X_1 \cup \cdots \cup X_{j-1}$ such that $(v + \d)(x) \le U_x$, the
time elapse $\d'$ still keeps $v'$ within the same integer. Clearly we
need this property to be satisfied so that $v'+\d'$ satisfies the same
LU-guards as $v + \d$. To this regard, we define the following two
values:
\begin{align*}
  l = \min & \{~\fracpart{v'(y)}~|~ (v+ \d)(y) \le L_y~\text{and}~y \in
  X_j
  \cup \cdots \cup X_k\} \\
  u = \max & \{~\fracpart{v'(x)}~|~ (v+\d)(x) \le U_x~\text{and}~x \in
  X_1 \cup \cdots \cup X_{j-1}\}
\end{align*}

Firstly note that $u < l$. If not, there exist clocks $y$ and $x$ such
that $v(y) \le L_y$ and $v(x) \le U_x$, for which the order property
is not true, thus giving a contradiction. Let
$\bar{\d}$ be a value between $u$ and $l$. Set $\d' = 1 -
\bar{\d}$. By construction, $v' + \d'$ satisfies the same LU-guards as
$v + \d$.

We will now see that this choice of $\d'$ also satisfies order property
for $v'+\d'$ and $v+\d$. Due to $\d'$ some clocks in $v'$ would have
crossed the next integer. Let us say that clocks in $X_{j'} \cup
\cdots \cup X_{k'}$ have crossed and the others stay within the same
integer. We pictorially depict the scenario with the two valuations $v
+ \d$ and $v' + \d'$ below.

\vspace{0.5cm}

\begin{center}
  \begin{tikzpicture}[scale = 1.3]

    \draw (0,0) -- (3,0); \fill (0,0) circle (1pt); \fill (0.3, 0)
    circle (1pt); \fill (1.1, 0) circle (1pt); \fill (1.6, 0) circle
    (1pt); \fill (2.5, 0) circle (1pt); \fill (3,0) circle (1pt);

    \node at (0, -0.2) {\scriptsize $0$}; \node at (0.3, -0.2)
    {\scriptsize $X_j$}; \node at (0.7, -0.2) {$\dots$}; \node at
    (1.1, -0.2) {\scriptsize $X_{k}$}; \node at (1.6, -0.2)
    {\scriptsize $X_1$}; \node at (2, -0.2) {$\dots$}; \node at (2.5,
    -0.2) {\scriptsize $X_{j-1}$}; \node at (3, -0.2) {\scriptsize
      $1$};

    \node at (1.5, -0.7) {\scriptsize In $v + \d$};

\begin{scope}[xshift=4.5cm]
  \draw (0,0) -- (3,0); \fill (0,0) circle (1pt); \fill (0.3, 0)
  circle (1pt); \fill (1.1, 0) circle (1pt); \fill (1.6, 0) circle
  (1pt); \fill (2.5, 0) circle (1pt); \fill (3,0) circle (1pt);

  \node at (0, -0.2) {\scriptsize $0$}; \node at (0.3, -0.2)
  {\scriptsize $X'_{j'}$}; \node at (0.7, -0.2) {$\dots$}; \node at
  (1.1, -0.2) {\scriptsize $X'_{k'}$}; \node at (1.6, -0.2)
  {\scriptsize $X'_1$}; \node at (2, -0.2) {$\dots$}; \node at (2.5,
  -0.2) {\scriptsize $X'_{j'-1}$}; \node at (3, -0.2) {\scriptsize
    $1$};

  \node at (1.5, -0.7) {\scriptsize In $v' + \d'$};

\end{scope}

\end{tikzpicture}
\end{center}

Pick two clocks $x,y$ such that:
\begin{align}\label{eq:4}
  \intpart{(v + \d)(x)} = \intpart{(v'+\d')(x)} & \text{ and }
  \intpart{(v+ \d)(y)} = \intpart{(v'+\d')(y)} \\
  (v+\d)(x) \le U_x & \text{ and } (v+\d)(y) \le L_y \nonumber \\
  \fracpart{(v+\d)(x)} & < \fracpart{(v+\d)(y)} \nonumber
\end{align}

Consider the case when both $x,y \in X_1 \cup \cdots \cup X_{j-1}$. As
they have not crossed integer in $v$, they should not have crossed
integer in $v'$ too because of (\ref{eq:4}). Therefore both $x,y \in
X'_1 \cup \cdots \cup X'_{j'-1}$. We know from order property that
$\fracpart{v'(x)} < \fracpart{v'(y)}$. Clearly the time elapse of
$\d'$ has not changed this ordering for these clocks and hence
$\fracpart{(v' + \d')(x)} < \fracpart{(v'+\d')(y)}$. We can prove
similarly when both $x, y \in X_{j} \cup \cdots \cup X_k$.

Let us now consider the case when $y \in X_1 \cup \cdots \cup X_{j-1}$
and $x \in X_j \cup \cdots \cup X_k$. As $x$ has crossed integer in
$v$, it should have crossed integer in $v'$ too by the hypothesis
(\ref{eq:4}). Therefore $x \in X'_{j'} \cup \cdots \cup
X'_{k'}$. Again by hypothesis (\ref{eq:4}) the clock $y$ should not
have crossed integer and hence $y \in X'_{1} \cup \cdots \cup
X'_{j'-1}$. Hence we get that $\fracpart{(v'+\d')(x)} <
\fracpart{(v'+\d')(y)}$.

This way we have proved the order property for $v+\d$ and
$v'+\d'$ for the case of the strict inequality. The case of the
equality can be handled in a similar way.
\end{proof}

\subsection{Finite paths characterizing LU-simulation}
\label{sec:finite-sequ-char}

The previous section took a digression to define the notion of
LU-regions. Now, we are in a position to answer the question: given
two valuations $v, v'$, when is $v \ftalu v'$. This section is devoted
to show the link between this question and the definition of
LU-regions. For valuations $v,v'$, we will show that $v \ftalu v'$ if
and only if $v'$ can elapse some amount of time and fall into the
LU-region of $v$ (Proposition~\ref{prop:rLU-sim}).  Before that, we
will define a sequence of guards that succinctly describes the
LU-region $\lureg{v}$.

\begin{definition}[LU-sequence]
  \label{def:lu-sequence}
  For a valuation $v$, let $g_{int}$ be the conjunction of all $LU$
  guards that $v$ satisfies. For every pair of clocks $x,y$ such that
  $v(x) \leq U_x$, $v(y) \leq L_y$, consider guards:
  \begin{itemize}
  \item if $0 < \fracpart{v(x)} < \fracpart{v(y)}$ then we take a
    guard $g_{xy}\equiv (x < \intpart{v(x)} + 1) \land (y >
    \intpart{v(y)} + 1)$.
  \item if $0 < \fracpart{v(x)} = \fracpart{v(y)}$ then we take a
    guard $g_{xy}\equiv (x \leq \intpart{v(x)} + 1) \land (y \geq
    \intpart{v(y)} + 1)$.
  \end{itemize}
  For every $y$ with $v(y) \le L_y$ put $g_y=\Land\set{g_{xy} : v(x)
    \leq U_x}$.  Consider all the clocks $y$ with $v(y) \leq L_y$ and
  suppose that $y_1,\dots,y_k$ is the ordering of these clocks with
  respect to the value of their fractional parts:
  $0\le\fracpart{v(y_1)}\leq\dots\leq\fracpart{v(y_k)}$.  The
  LU-sequence $\mathit{seq}(v)$ is defined to be the sequence of
  transitions $\act{g_{int}}\ \act{g_{y_k}}\ \dots\ \act{g_{y_1}}$
\end{definition}

\begin{proposition}
  \label{prop:rLU-sim}
  For every two valuations $v$ and $v'$:
  \begin{equation*}
    v\ftalu v'\quad\text{iff}\quad \text{there is $\d'\in\Rpos$ with
      $v'+\d'\in \rrlu{v}$} .
  \end{equation*}
\end{proposition}
\begin{proof}
  First let us take $v$ and consider its LU-sequence $\seq(v)$.  The
  sequence $\seq(v)$ can be performed from $v$ (the symbol $\tau$
  denotes a time elapse):
  \begin{multline*}
    v \xra{g_{int}} v \xra{\tau} v+\d_k \xra{g_{y_k}} v+\d_k
    \xra{\tau} v+\d_{k-1} \xra{g_{y_{k-1}}} \dots \\
    \dots \xra{\tau} v+\d_1 \xra{g_{y_1}} v+\d_1
  \end{multline*}
  when choosing $\d_i = (1 - \fracpart{v(y_i)})$ or $\d_i = (1 -
  \fracpart{v(y_i)}) + \varepsilon$ for some sufficiently small
  $\varepsilon > 0$; depending on whether we test for non-strict or
  strict inequality in $g_{y_i}$. Delay $\d_i$ makes the value of
  $y_i$ integer or just above integer.

  If $v \ftalu v'$, then there exists a $\d'$ such that $v' + \d'$ can
  do the sequence of transitions given by $\mathit{seq}(v)$. The guard
  $g_{int}$ ensures that $v'+\d'$ satisfies the same $LU$-guards as
  $v$. Note that in particular, this entails that for every pair of
  clocks $x,y$ such that $v(x) \le U_x$, $v(y) \le L_y$ and
  $\fracpart{v(x)} > 0$, we have:
  \begin{itemize}
  \item $\intpart{(v'+\d')(x)} < \intpart{v(x)} + 1$, and
  \item $\intpart{(v'+\d')(y)} \ge \intpart{v(y)}$.
  \end{itemize}
  Following transition $g_{int}$, valuation $v'+\d'$ can satisfy
  guards $g_{y_k}$ to $g_{y_1}$ by letting some time elapse:
  \begin{multline*}
    v'+\d' \xra{g_{int}} v'+\d' \xra{\tau} v+\d'_k \xra{g_{y_k}}
    v+\d'_k
    \xra{\tau} v+\d'_{k-1} \xra{g_{y_{k-1}}} \dots \\
    \dots \xra{\tau} v'+\d_1 \xra{g_{y_1}} v'+\d'_1
  \end{multline*}
  
  Consider the clock $y_i$. If the integral part
  $\intpart{(v'+\d')(y_i)}$ is strictly greater than $\intpart{v(y_i)}$,
  time elapse is not necessary to cross the guard $g_{y_i}$. On the
  other hand, if $\intpart{(v'+\d')(y_i)} = \intpart{v(y_i)}$, then for
  the guard $g_{y_i}$ to be crossed, $\d'_i$ should be sufficiently
  large to make the value of $(v' + \d'_i)(y_i)$ integer or just above
  integer. But at the same time, the guard $g_{xy_i}$ is satisfied,
  which entails that for all $x$ such that $v(x) \le U_x$,
  $\intpart{v(x)} = \intpart{(v'+\d')(x)}$, we get:
  \begin{itemize}
  \item if $0 < \fracpart{v(x)} < \fracpart{v(y_i)}$, then
    $\fracpart{(v'+\d')(x)} < \fracpart{(v'+\d')(y_i)}$ and
  \item if $0 < \fracpart{v(x)} = \fracpart{v(y_i)}$, then
    $\fracpart{(v'+\d')(x)} \le \fracpart{(v'+\d')(y_i)}$
  \end{itemize}
  Therefore, from the definition of $LU$-regions, we get that $v' +
  \d' \in \rrlu{v}$.  This shows left to right implication.

  For the right to left implication we show that the relation
  $S=\set{(v,v') : v'\in \rrlu{v}}$ is an LU-simulation relation. For
  this we take any $(v,v')\in S$, any $LU$ guard $g$, and any reset
  $R$ such that $v\act{g,R} v_1$. We show that $v'\act{g,R} v'_1$ for
  some $v'_1$ with $(v_1,v'_1)\in S$. The only non-trivial part in
  this is to show that if $v + \d \sat g$ for some $\d$, then there
  exists a $\d'$ such that $(v+\d, v' + \d') \in S$ and $v' + \d' \sat
  g$. But this is exactly given by Lemma~\ref{lem:lu-reg-time}.
\end{proof}

In particular the proof shows the following.
\begin{corollary}\label{cor:sequence}
  For two valuations $v$, $v'$:
  \begin{equation*}
    \text{$v \ftalu v'$ iff $v'$ can execute the sequence
      $\seq(v)$.}
  \end{equation*}
\end{corollary}

\subsection{Proof of optimality}
\label{sec:proof-optimality}

We are now ready to 
show that $\abslu(Z)$ (Definition \ref{def:abslu}), the abstraction
based on $\ftalu$ simulation, is the coarsest
sound and complete abstraction that uses solely the $LU$ information.

\begin{theorem}
  \label{thm:alu-biggest}
  The $\abslu$ abstraction is the coarsest abstraction that is sound
  and complete for all LU-automata. It is also finite.
\end{theorem}

\begin{proof}
  Suppose that we have some other abstraction $\aaa'$ that is not
  included in $\abslu$ on at least one $LU$-automaton. This means that
  there is some $LU$ automaton $\Aa_1$ and its reachable configuration
  $(q_1,Z)$ such that $\aaa'(Z)\setminus\abslu(Z)$ is not empty. We
  suppose that $\aaa'$ is complete and show that it is not sound.
  
  Take $v\in \aaa'(Z)\setminus\abslu(Z)$. Consider the test sequence
  $\seq(v)$ as in Corollary~\ref{cor:sequence}. From this corollary we
  know that it is possible to execute this sequence from $v$ but it is
  not possible to do it from any valuation in $Z$ since otherwise we
  would get $v\in \abslu(Z)$.

  As illustrated in Fig~\ref{fig:a1-seq} we add to $\Aa_1$ a new
  sequence of transitions constructed from the sequence $\seq(v)$.  We
  start this sequence from $q_1$, and let $q_f$ be the final state of
  this new sequence. The modified automaton $\Aa_1$ started in the
  initial configuration arrives with $(q_1,Z)$ in $q_1$ and then it
  can try to execute the sequence we have added. From what we have
  observed above, it will not manage to reach $q_f$. On the other hand
  from $(q_1,v)$ it will manage to complete the sequence. But then by
  completeness of the abstraction $(q_1,\aaa'(Z))\act{\seq(v)}
  (q_f,W)$ for a nonempty $W$. So $\aaa'$ is not a sound abstraction.

  That $\abslu$ is finite is easy to see. The set $\abslu(Z)$ is a
  union of classical regions. Recall that we denote by $M$ the bound function
  that assigns to each clock $x$, the maximum of $L_x$ and $U_x$. Let
  $v'$ be a valuation in $Z$. If $v' \in \reg{v}$ then
  it is easy to see that $v' \in \rrlu{v}$ and by definition $v' \in
  \abslu(Z)$.
\end{proof}

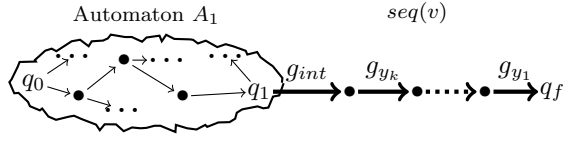
\begin{figure}[!t]
  \centering \input{fig-a1-seq}
  \caption{Adding the sequence $\seq(v)$ to $\Aa_1$.}
  \label{fig:a1-seq}
\end{figure}


%% file: fig-a1-seq.tex
{
  \begin{tikzpicture}[scale=0.6]
    \path[draw,thick,decorate,
    decoration={random steps,segment length=3pt,amplitude=2pt}] (0,0)
    ellipse (3cm and 1cm);
    \begin{scope}[inner sep=0.5pt]
      \node (q0) at (-2.5,0) {$q_0$};
      \node (q0s1) at (-1.6,0.6) {\bfseries\dots};
      \node (q0s2) at (-1.5,-0.3) {$\bullet$};
      \node (q0s2s1) at (-0.5,-0.6) {\bfseries\dots};
      \node (q0s2s2) at (-0.5,0.5) {$\bullet$};
      \node (q0s2s2s1) at (0.5,0.5) {\bfseries\dots};
      \node (q0s2s2s2) at (0.8,-0.3) {$\bullet$};
      \node (q1) at (2.5,-0.2) {$q_1$};
      \node (q1s1) at (1.8,0.6) {\bfseries\dots};
    \end{scope}
    \begin{scope}[->]
      \draw(q0) -- (q0s1);
      \draw(q0) -- (q0s2);
      \draw(q0s2) -- (q0s2s1);
      \draw(q0s2) -- (q0s2s2);
      \draw(q0s2s2) -- (q0s2s2s1);
      \draw(q0s2s2) -- (q0s2s2s2);
      \draw(q0s2s2s2) -- (q1);
      \draw (q1) -- (q1s1);
    \end{scope}
    \node[above] (A1) at (0,1.1) {\footnotesize Automaton $A_1$};
    \begin{scope}[inner sep=0.5pt]
      \node (gint) at (4.5,-0.2) {$\bullet$};
      \node (qyk) at (6,-0.2) {$\bullet$};
      \node (qy1) at (7.5,-0.2) {$\bullet$};
      \node (qf) at (9,-0.2) {$q_f$};
    \end{scope}
    \begin{scope}[->,ultra thick]
      \draw (q1) -- node[above] {$g_{int}$} (gint);
      \draw (gint) -- node[above] {$g_{y_k}$} (qyk);
      \draw[dotted] (qyk) -- (qy1);
      \draw (qy1) -- node[above] {$g_{y_1}$} (qf);
    \end{scope}
    \node[above] (seqv) at (6, 1.1) {\footnotesize $seq(v)$};
  \end{tikzpicture}
}


%% file: alu-abslu.tex
\section{The $\alu$ abstraction}
\label{sec:abstr-alu-coinc}

Since $\abslu$ is the coarsest abstraction, we would like to use it in
a reachability algorithm. The definition of $\abslu$, or even the
characterization referring to LU-regions, are still too complicated to
work with. It turns out though that there is a close link to an
existing abstraction.

The $\alu$ abstraction proposed by Behrmann et al. in
\cite{Behrmann:STTT:2006} has a much simpler definition. Quite
surprisingly, in the context of reachability analysis the two
abstractions coincide (Theorem~\ref{thm:alu_biggest}).
The $\alu$ abstraction is based on a
simulation relation $\lu$. 

\begin{definition}[LU-preorder~\cite{Behrmann:STTT:2006}]
  \label{def:lu-preorder}
  Let $L,U:X\to\Nat \cup \{-\infty\}$ be two bound functions. For a
  pair of valuations we set $v\lu v'$ if for every clock $x$:
  \begin{itemize}
  \item if $v'(x)<v(x)$ then $v'(x)> L_x$, and
  \item if $v'(x)>v(x)$ then $v(x)>U_x$.
  \end{itemize}
\end{definition}

It has been shown in~\cite{Behrmann:STTT:2006} that $\lu$ is a
simulation relation.  The $\alu$ abstraction is based on this
relation.

\begin{definition}[$\alu$-abstraction~\cite{Behrmann:STTT:2006}]
  Given $L$ and $U$ bound functions, for a set of valuations $W$ we
  define:
  \begin{equation*}
    \alu(W)=\set{v ~|~ \exists v'\in W.\  v\lu v'}.
  \end{equation*}
\end{definition}

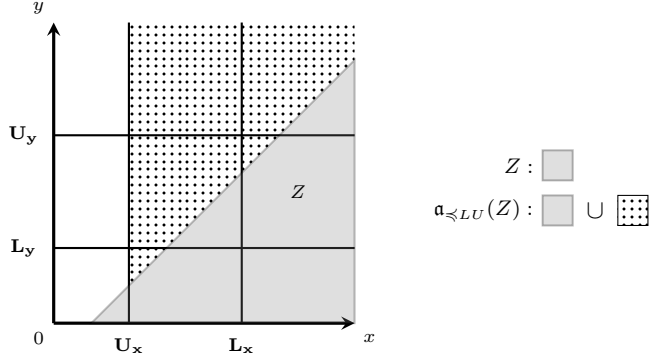
\begin{figure}[!t]
  \centering \input{fig/alu-example.tex}
  \caption{Zone $Z$ is given by the grey area. Abstraction $\alu(Z)$
    is given by the grey area along with the dotted area}
  \label{fig:alu-example}
\end{figure}

Figure~\ref{fig:alu-example} gives an example of a zone $Z$ and its
abstraction $\alu(Z)$. It can be seen that $\alu(Z)$ is not a convex
set.

\subsection{Abstractions $\abslu$ and $\alu$ coincide}

Our goal is to show that when we consider zones closed under
time-successors, $\alu$ and $\abslu$ coincide. To prove this, we would
first show that there is a very close connection between valuations in
$\rrlu{v}$ and valuations that simulate $v$ with respect to $\lu$.  The
following lemma says that if $v' \in \rrlu{v}$ then by slightly
adjusting the fractional parts of $v'$ we can get a valuation $v'_1$
such that $v \lu v'_1$. We start with a preliminary definition.

\begin{definition}
  A valuation $v_1$ is said to be in the \emph{neighborhood} of $v$,
  written $v_1 \in \nbd{v}$ if for all clocks $x,y$:
  \begin{itemize}
  \item $\intpart{v(x)} = \intpart{v_1(x)}$,
  \item $\fracpart{v(x)} = 0$ iff $\fracpart{v_1(x)} = 0$,
  \item $\fracpart{v(x)} \lleq \fracpart{v(y)}$ implies
    $\fracpart{v_1(x)} \lleq \fracpart{v_1(y)}$ where $\lleq$ is either
    $<$ or $=$.
  \end{itemize}
\end{definition}

Notice that the neighborhood of $v$ is the same as the region of $v$
with respect to the classical region definition
with maximal bound being $\infty$.

We give a brief intuition before proving the following lemma which
gives the relation between LU-regions and $\lu$ simulation. Consider a
valuation $v$ shown in Figure~\ref{fig:adjustment}. Its LU-region
$\lureg{v}$ is given by the shaded portion. Pick a valuation $v'$ that
belongs to $\lureg{v}$ and let us see if it satisfies $v \lu v'$.

As we can see in Figure~\ref{fig:adjustment}, the value of $v'(x) >
v(x)$ and additionally $v'(x) \le U_x$. This shows that $v \not \lu
v'$ due to its $x$-coordinate. However, by slightly adjusting the
fractional parts: that is, reducing $\fracpart{v'(y)}$ to move it down
a bit and then reducing $\fracpart{v'(x)}$ to make $v'(x)$ equal to
$v(x)$ leads us to a valuation $v_1'$ which is in the neighborhood of
$v'$ but now $v \lu v_1'$. The adjustment is depicted in
Figure~\ref{fig:adjustment}.  Essentially, the following lemma claims
that the LU-region $\lureg{v}$ can be obtained as the downward closure
of $\lu$ over the set $\nbd{v}$, in other words, $\alu(\nbd{v})$. 

\begin{figure}
 \centering
 \input{fig/adjustment}
 \caption{Intuition for the adjustment lemma.}
\label{fig:adjustment}
\end{figure}
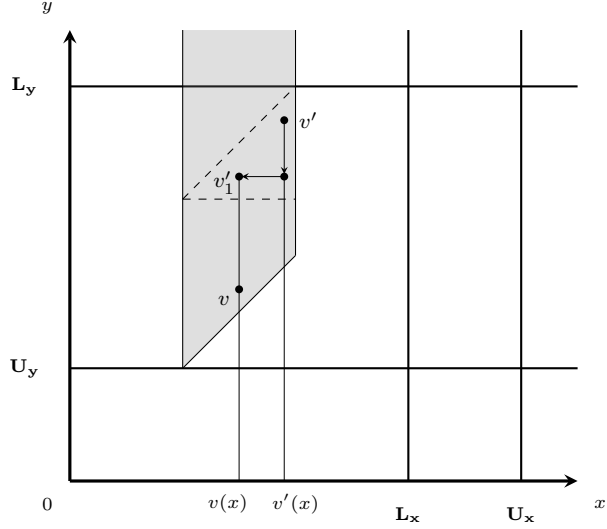

\begin{lemma}[Adjustment]\label{lem:adjustment}
  Let $v$ be a valuation and let $v' \in \rrlu{v}$. Then, there exists
  a $v_1' \in \nbd{v'}$ such that $v \lu v_1'$.
\end{lemma}

\begin{proof}
  Let $v' \in \rrlu{v}$. The goal is to construct a valuation $v'_1 \in
  \nbd{v'}$ that satisfies $v \lu v'_1$. To be in the neighborhood,
  the valuation $v'_1$ should have the same integral parts as that of
  $v'$ and should agree on the ordering of fractional parts. So for
  all $x$, we put $\intpart{v'_1(x)} = \intpart{v'(x)}$. It remains to
  choose the fractional parts for $v'_1$. But before, we will first
  see that there are clocks for which irrespective of what the
  fractional part is, the two conditions in Definition
  \ref{def:lu-preorder} would be true.

  Consider a clock $x$ that has $\intpart{v'(x)} <
  \intpart{v(x)}$. Since $v'$ satisfies all LU-guards as $v$, we
  should have $v'(x) > L_x$. The first condition of $\lu$ for $x$
  becomes true and the second condition is vacuously true.  Similarly,
  when $\intpart{v'(x)} > \intpart{v(x)}$, we should have $v(x) > U_x$
  and the second condition of $\lu$ becomes true and the first
  condition is vacuously true. Therefore, clocks $x$ that do not have
  the same integral part in $v$ and $v'$ satisfy the $\lu$ condition
  directly thanks to the different integral parts. Whatever the
  fractional parts of $v'_1$ are, the $\lu$ condition for these clocks
  would still be true.

  Let us therefore now consider only the clocks that have the same
  integral parts: $\intpart{v'(x)} = \intpart{v(x)}$. If this integer
  is strictly greater than both $L_x$ and $U_x$, the two conditions of
  $\lu$ would clearly be satisfied, again irrespective of the
  fractional parts. So we consider only the clocks $x$ that have the
  same integral part in both $v$ and $v'$ and additionally either
  $\intpart{v(x)} \le U_x$ or $\intpart{v(x)} \le L_x$.

  We prune further from among these clocks. Suppose there is such a
  clock that has $\fracpart{v'(x)} = 0$. To be in the neighborhood,
  we need to set $\fracpart{v'_1(x)} = 0$. If $\fracpart{v(x)}$ is $0$
  too, we are done as the $\lu$ condition becomes vacuously
  true. Otherwise, we would have $v'(x) = v'_1(x) < v(x)$. But recall
  that $v' \in \rrlu{v}$ and so it satisfies the same LU-guards as $v$
  does. This entails that $v'_1(x) > L_x$ and we get the first
  condition of $\lu$ to be true. Once again, the other condition is
  trivial. So we eliminate clocks that have zero fractional parts in
  $v'$. A similar argument can be used to eliminate clocks that have
  zero fractional parts in $v$.

  So finally, we end up with the set of clocks $x$ that have:
  \begin{itemize}
  \item $\intpart{v'(x)} = \intpart{v(x)}$,
  \item $\fracpart{v'(x)} > 0$ and $\fracpart{v(x)} > 0$,
  \item $v(x) < \max(U_x, L_x)$.
  \end{itemize}
  
  Call this set $X_{f}$.  The task is to select non-zero fractional
  values $\fracpart{v'_1(x)}$ for all clocks $x$ in $X_f$ so that they
  match with the order in $v'$. This is the main challenge and this is
  where we would be using the second property in the definition of $v'
  \in \rrlu{v}$, which we restate here:
  \begin{align}\label{eqn:rlu-prop2}
    & \text{$\forall x, y \in X_f$ such that $v(x) \le U_x$ and $v(y)
      \le
      L_y$}  \\
    & \qquad 0 < \fracpart{v(x)} < \fracpart{v(y)} \imp
    \fracpart{v'(x)} < \fracpart{v'(y)} \nonumber \\
    & \qquad 0 < \fracpart{v(x)} = \fracpart{v(y)} \imp
    \fracpart{v'(x)} \le \fracpart{v'(y)} \nonumber
  \end{align}

  Let $0 < \l_1' < \l_2' < \dots < \l_n' < 1$ be the fractional values
  taken by clocks of $X_f$ in $v'$, that is, for every clock $x \in
  X_f$, the fractional value $\fracpart{v'(x)} = \l'_i$ for some $i
  \in \{1, \dots, n \}$. Let $X_i$ be the set of clocks $x \in X_f$
  that have the fractional value as $\l'_i$:
  \begin{align*}
    X_i = \{ x \in X_f~|~ \fracpart{v'(x)} = \l'_i\}
  \end{align*}
  for $i \in \{1, \dots, n\}$.

  In order to match with the ordering of $v'$, one can see that for
  all clocks $x_i$ in some $X_i$, the value of $\fracpart{v'_1(x_i)}$
  should be the same, and if $x_j \in X_j$ with $i \neq j$, then we
  need to choose $\fracpart{v'_1(x_i)}$ and $\fracpart{v'_1(x_j)}$
  depending on the order between $\l'_i$ and $\l'_j$.

  Therefore, we need to pick $n$ values $0 < \s_1 < \s_2 < \dots <
  \s_n < 1$ and assign for all $x_i \in X_i$, the fractional part
  $\fracpart{v'_1(x_i)} = \s_i$. We show that it can be done by an
  induction involving $n$ steps.

  After the $k^{th}$ step of the induction we assume the following
  hypothesis:
  \begin{itemize}
  \item we have picked values $0 < \s_{n - k+1} < \s_{n -k+2} < \dots
    < \s_n < 1$,
  \item for all clocks $x \in X_{n -k +1} \cup X_{n -k +2} \dots \cup
    X_n$, the $\lu$ condition is satisfied,
  \item for all clocks $y \in X_1 \cup X_2 \dots \cup X_{n-k}$, we
    have
    \begin{align}
      v(y) \le L_y \imp \fracpart{v(y)} < \s_{n - k +1}
    \end{align}
  \end{itemize}
  
  Let us now perform the $k+1^{th}$ step and show that the induction hypothesis
  is true for $k+1$. The task is to pick $\s_{n - k}$.  We first
  define two values $0 < l < 1$ and $0 < u < 1$ as follows:
  \begin{align*}
    l & = \max\big\{\,\fracpart{v(z)}~|~z \in X_{n-k} \text{ and }
    v(z)\leq L_z\,\big\} \\
    u & = \min\big\{\,\{\,\fracpart{v(z)}~|~z \in X_{n-k} \text{ and }
    v(z) \leq U_z\,\} \cup \s_{n-k+1}\,\big\}
  \end{align*}
  We claim that $l\le u$. Firstly, $l < \s_{n-k+1}$ from the third
  part of the induction hypothesis. So if $u$ is $\s_{n-k+1}$ we are
  done. If not, suppose $l > u$, this means that there are clocks $x,
  y \in X^{n-k}$ with $v(x) \leq U_x$ and $v(y) \leq L_y$ such that
  $\fracpart{v(x)} < \fracpart{v(y)}$. From Equation
  \ref{eqn:rlu-prop2}, this would imply that $\fracpart{v'(x)} <
  \fracpart{v'(y)}$. But this leads to a contradiction since we know
  they both equal $\l'_{n-k}$ in $v'$.

  This leaves us with two cases, either $l = u$ or $l < u$. When $l =
  u$, we pick $\s_{n-k} = l = u$. Firstly, from the third part of the
  hypothesis, we should have $l < \s_{n-k+1}$ and so $\s_{n-k} <
  \s_{n-k+1}$. Secondly for all $z \in X_{n-k}$, if $v'_1(z) < v(z)$,
  then $z$ should not contribute to $l$ and so $v(z) > L_z$, which is
  equivalent to saying, $v'_1(z) > L_z$. Similarly, if $v'_1(z) >
  v(z)$, then $z$ should not contribute to $u$ and so $v(z) > U_z$,
  thus satisfying the $\lu$ condition for $z$. Finally, we should show
  the third hypothesis. Consider a clock $y \in X_1 \cup \dots \cup
  X_{n-k-1}$ with $v(y) < L_y$. If $\fracpart{v(y)} \ge \s_{n-k}$, it
  would mean that $\fracpart{v(y)} \ge u$ and from Equation
  \ref{eqn:rlu-prop2} gives a contradiction. So the three requirements
  of the induction assumption are satisfied after this step in this
  case.

  Now suppose $l < u$. Consider a clock $y \in X_1 \cup \dots \cup
  X_{n-k-1}$ such that $v(y) < L_y$. From Equation
  \ref{eqn:rlu-prop2}, we should have $\fracpart{v(y)} < u$. Take the
  maximum of $\fracpart{v(y)}$ over all such clocks:
  \begin{align*}
    \l = \max\{ \fracpart{v(y)}~|~ y \in X_1 \cup \dots \cup X_{n-k-1}
    \text{ and } v(y) < L_y \}
  \end{align*}
  Choose $\s_{n-k}$ in the interval $(\l, u)$.  We can see that all
  the three assumptions of the induction hold after this step.
  
\end{proof}

We can now prove the main result of this section. Recall
Definition~\ref{def:time-elapsed-zone} of time-elapsed zones. 

\begin{theorem}
  \label{thm:alu_biggest}
  If $Z$ is time-elapsed then
  \begin{equation*}
    \abslu(Z) = \alu(Z)
  \end{equation*}
\end{theorem}

\begin{proof}
  Suppose $v \in \alu(Z)$. There exists a $v' \in Z$ such that $v \lu
  v'$. It can be easily verified that $\lu$ is a $LU$-simulation
  relation. Since $\ftalu$ is the biggest LU-simulation, we get that
  $v \ftalu v'$. Hence $v \in \abslu(Z)$.

  Suppose $v \in \abslu(Z)$. There exists $v' \in Z$ such that $v
  \ftalu v'$. From Proposition~\ref{prop:rLU-sim}, this implies that there exists
  a $\d'$ such that $v' + \d' \in \rrlu{v}$. As $Z$ is time-elapsed, we
  get $v' + \d' \in Z$. Moreover, from Lemma~\ref{lem:adjustment}, we
  know that there is a valuation $v'_1 \in \nbd{v' + \d'}$ such that
  $v \lu v'_1$. Every valuation in the neighborhood of $v'+\d'$
  satisfies the same constraints of the form $y - x \lleq c$ with
  respect to all clocks $x,y$ and hence $v'_1$ belongs to $Z$
  too. Therefore, we have a valuation $v'_1 \in Z$ such that $v \lu
  v'_1$ and hence $v \in \alu(Z)$.
\end{proof}

\subsection{Using $\alu$ to solve the reachability problem}
\label{sec:using-alu-solve}

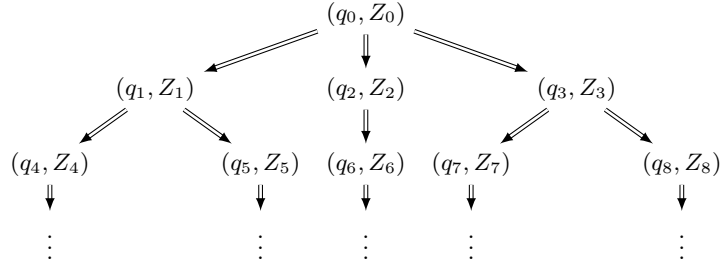
\begin{figure}[!t]
  \centering \input{fig-reach-tree}
  \caption{A reachability tree of zones computed by a forward
    exploration}
  \label{fig:reach-tree}
\end{figure}

A forward exploration algorithm for solving the reachability problem
constructs the reachability tree starting from the initial node
$(q_0,Z_0)$ (cf. Figure~\ref{fig:reach-tree}), with
$Z_0 = \{ \vali + \d~|~ \d \in \Rpos\}$. The successor with respect to
$\tto$ can be computed in time $\Oo(|X|^2)$ where $X$ is the number of
clocks~\cite{Zhao:IPL:2005}. By definition $\tto$ computes a
time-elapsed zone. Therefore, all nodes that are explored by the
algorithm have time-elapsed zones.

Before continuing exploration from a node $(q,Z)$, the algorithm first
checks if $q$ is accepting. If not, the algorithm checks if for some
visited node $(q,Z')$, we have $Z \incl \alu(Z')$. If this is the
case, $(q,Z)$ need not be explored. Otherwise, the successors of
$(q,Z)$ are computed as stated above. This way we ensure termination
of the algorithm since $\alu$ is a finite
abstraction~\cite{Behrmann:STTT:2006}.

Since the reachability algorithm refers to only time-elapsed zones,
Theorems~\ref{thm:alu-biggest} and~\ref{thm:alu_biggest} show that
$\alu$ is the coarsest sound and complete abstraction provided the only
thing we know about the structure of the automaton are its $L$ and $U$
bounds. Recall that coarser abstractions make abstract graph smaller,
so the exploration algorithm can finish faster.

In Definition~\ref{def:LU-bounds}, we introduced $LU$-bounds that
associate an $L$ bound and a $U$ bound to every clock in an automaton
$\Aa$. Those bounds are the same in every state of the
automaton. Instead, state-of-the-art algorithms calculate $LU$-bounds
for each state of the automaton separately~\cite{Behrmann:TACAS:2003},
or even on-the-fly during
exploration~\cite{Herbreteau:FSTTCS:2011,CAV}. The maximality argument
in favor of $\alu$ is of course true also in this case.

The last missing piece is an efficient inclusion test $Z\incl
\alu(Z')$. This is the main technical contribution of this paper.


%% file: fig/alu-example.tex
\begin{tikzpicture}

    
  \begin{scope}[xshift = -1cm]
    \draw[thick] (2,0) -- (2,4); \draw (2, -0.3) node {\scriptsize
      $\mathbf{U_x}$} (2.5, -0.4); \draw[thick] (3.5,0) --
    (3.5,4); \draw (3.5, -0.3) node {\scriptsize
     $\mathbf{L_x}$} (4, -0.4);
  \end{scope}
  
  \begin{scope}[yshift=-0.5cm]
    \draw[thick] (0,1.5) -- (4,1.5); \draw (-0.4, 1.5) node
    {\scriptsize $\mathbf{L_y}$} (0,1.5); \draw[thick]
    (0,3) -- (4,3); \draw (-0.4, 3) node {\scriptsize
      $\mathbf{U_y}$} (0, 3);
  \end{scope}

  \draw[->, very thick, >=stealth] (0,0) -- (0,4); \draw[->, very
  thick, >=stealth] (0,0) -- (4,0); \draw
  (-0.2,-0.2) node {\scriptsize $0$} (0, -0.4); \draw (4.2, -0.2) node
  {\scriptsize $x$} (4.2, -0.2); \draw (-0.2, 4.2) node {\scriptsize
    $y$} (0,4.2);

  \draw[thick, fill=gray, nearly transparent] (0.5,0)
    -- (4,3.5) -- (4,0) -- cycle;
  \draw (3.25, 1.75) node {\scriptsize $Z$ } (3.5, 1.75); 

 \fill[pattern = dots] (1,0.5) -- (1,4) -- (4,4) -- (4,3.5) -- cycle; 



 \begin{scope}[xshift = -3cm, yshift=-1cm]
    \node[left] at (9.5, 3.1) {\footnotesize $Z :$};
    \node[left] at (9.5, 2.5) {\footnotesize $\alu(Z):$};

    \draw[thick, fill=gray, nearly transparent] (9.5,2.9) rectangle
   (9.9, 3.3);
    \draw[thick, fill=gray, nearly transparent] (9.5,2.3) rectangle
  (9.9, 2.7);
   
   \node at (10.2, 2.5) {$\cup$};
   \draw[fill, pattern = dots] (10.5,2.3) rectangle
   (10.9, 2.7);
\end{scope}

\end{tikzpicture}


%% file: fig/adjustment.tex
\begin{tikzpicture}[scale = 1.5]
  \begin{scope}[xshift=1cm]
    \draw[thick] (2,0) -- (2,4); \draw (2, -0.3) node {\scriptsize
      $\mathbf{L_x}$} (2.5, -0.4); \draw[thick] (3,0) --
    (3,4); \draw (3, -0.3) node {\scriptsize
     $\mathbf{U_x}$} (3.5, -0.4);
  \end{scope}
  
  \begin{scope}[yshift=-0.5cm]
    \draw[thick] (0,1.5) -- (4.5,1.5); \draw (-0.4, 1.5) node
    {\scriptsize $\mathbf{U_y}$} (0,1.5); \draw[thick]
    (0,4) -- (4.5,4); \draw (-0.4, 4) node {\scriptsize
      $\mathbf{L_y}$} (0, 4);
  \end{scope}

  \draw[->, very thick, >=stealth] (0,0) -- (0,4);
  \draw[->, very thick, >=stealth] (0,0) -- (4.5,0); \draw
  (-0.2,-0.2) node {\scriptsize $0$} (0, -0.4); \draw (4.7, -0.2) node
  {\scriptsize $x$} (4.2, -0.2); \draw (-0.2, 4.2) node {\scriptsize
    $y$} (0,4.2);

 \fill[gray, nearly transparent] (1,1) -- (1,4) -- (2,4) -- (2,2) -- cycle; 
  
 \draw (1,1) -- (1,4);
 \draw (2,2) -- (2,4);
 \draw (1,1) -- (2,2);

 \fill (1.5, 1.7) circle (1pt);
 \node[left] at (1.5, 1.6) {\footnotesize $v$};

 \draw [dashed] (1, 2.5) -- (2, 3.5);
 
 \fill (1.5, 2.7) circle (1pt);
 \node[left] at (1.55, 2.65) {\footnotesize $v_1'$};
 
 \fill (1.9, 2.7) circle (1pt);
 \node[right] at (1.95, 3.2) {\footnotesize $v'$};
 \fill (1.9, 3.2) circle (1pt);
 
 \draw [dashed] (1, 2.5) -- (2, 2.5);
 \draw[very thin, ->, >=stealth, shorten >= 1pt] (1.9, 2.7) -- (1.5,
 2.7); 
 \draw[very thin, ->, >=stealth, shorten >= 1pt] (1.9, 3.2) -- (1.9,
 2.7); 
 
 \draw [very thin] (1.5, 0) -- (1.5, 2.7);
 \node at (1.4, -0.2) {\scriptsize $v(x)$};
 
 \draw [very thin] (1.9, 0) -- (1.9, 2.7);
 \node at (2, -0.2) {\scriptsize $v'(x)$};
\end{tikzpicture}


%% file: fig-reach-tree.tex
{\small
  \begin{tikzpicture}[level distance=1cm, sibling distance=2.8cm,
    double distance=1pt,>=latex] 
    \node {$(q_0, Z_0)$}
    child { node {$(q_1, Z_1)$} edge from parent[double,->]
      child { node {$(q_4, Z_4)$} edge from parent[double,->] 
        child { node {$\vdots$} edge from parent[double,->] }
      }
      child { node {$(q_5, Z_5)$} edge from parent[double,->]
         child { node {$\vdots$} edge from parent[double,->] }
       }
    }
    child { node {$(q_2, Z_2)$} edge from parent[double,->] 
       child { node {$(q_6, Z_6)$} edge from parent[double,->] 
        child { node {$\vdots$} edge from parent[double,->] }
      }
    }
    child { node {$(q_3, Z_3)$} edge from parent[double,->]
      child { node {$(q_7, Z_7)$} edge from parent[double,->] 
        child { node {$\vdots$} edge from parent[double,->] }
      }
      child { node {$(q_8, Z_8)$} edge from parent[double,->]
         child { node {$\vdots$} edge from parent[double,->] }
       }
    };
  \end{tikzpicture}
}


%% file: inclusion.tex
\section{An $\Oo(|X|^2)$ algorithm for $Z \incl \alu(Z')$}
\label{sec:an-oox2-algorithm}

In this section, we present an efficient algorithm for the inclusion
test $Z \incl \alu(Z')$ (Theorem~\ref{thm:effic-inc-test}). In the
algorithm for the reachability problem as explained in
Section~\ref{sec:using-alu-solve}, each time a new node $(q,Z)$ is
seen, it is checked if there exists an already visited node $(q, Z')$ such that
$Z \incl \alu(Z')$. This means that a lot of such inclusion tests
need to be performed during the course of the algorithm. Hence it is
essential to have a low complexity for this inclusion procedure. We
are aiming at quadratic complexity as this is the complexity incurred
in the existing algorithms for inclusions of the form $Z \incl Z'$
used in the standard reachability algorithm.  It is well known that
all the other operations needed for forward exploration, can be done
in at most quadratic time~\cite{Zhao:IPL:2005}.

We will now characterize when $Z \not \incl \alu(Z')$ holds. The main
steps are outlined below.

\begin{description}

\item \textbf{Step 1}. As a first step, we reduce the inclusion
  problem to a problem of intersection in
  Section~\ref{sec:incl-inters}. The question $Z \not \incl \alu(Z')$
  boils down to asking if there exists a valuation $v \in Z$ such that
  its LU-region $\lureg{v}$ does not intersect $Z'$
  (Proposition~\ref{prop:incl-int}). 

\item \textbf{Step 2}. As a next step, we consider the intersection
  $\lureg{v} \cap Z'$. We aim to show that this intersection can be
  decided by looking at projections on every pair of clocks
  (Proposition~\ref{prop:neg-cycles-in-min}). This is the most
  difficult step in the way to the inclusion test and spans three
  sections.  We first describe a convenient graph representation of
  zones in Section~\ref{sec:distance-graphs}. We call this the
  distance graph and will use it to represent $Z'$. Subsequently, in
  Section~\ref{sec:lu-regions-as}, we see how we can represent
  LU-regions as distance graphs. This gives a distance graph
  representation for $\lureg{v}$.  Finally, in
  Section~\ref{sec:when-does-an}, we analyze the graph representations
  of $\lureg{v}$ and $Z'$ to see when the intersection $\lureg{v} \cap
  Z'$ is empty. We show that for this it is enough to look at two clocks at a
  time. 

\item \textbf{Step 3}. The previous step gives the condition for
  $\lureg{v} \cap Z'$ to be empty. We now look at the zone $Z$ to find
  out quickly the valuation $v \in Z$ that can potentially satisfy this
  condition (Proposition~\ref{prop:minimum-val-of-xy}).

\item \textbf{Step 4}. We substitute the valuation obtained from Step
  3 to the condition of Step 2 to give the efficient test for
  inclusion (Theorem~\ref{thm:effic-inc-test}). Both steps 3 and 4 appear in
  Section~\ref{sec:final-step}.
\end{description}

\paragraph*{Notation} For notational convenience, we denote $v(x)$ by $v_x$ for a valuation
$v$ and clock $x$.

\subsection{Reducing inclusion to intersection}
\label{sec:incl-inters}

The aim of this chapter is to reduce the question of inclusion to a
question of intersection.  The adjustment lemma
(Lemma~\ref{lem:adjustment}) shows a close connection between
$LU$-regions and $\lu$-simulation in one direction: that is, if $v'
\in \erlu{v}$ then we can find a valuation $v_1'$ in the neighborhood
of $v'$ such that $v \lu v_1'$. We show below a connection in the
other direction too.

\begin{lemma}\label{lem:lu-imp-rlu}
  Let $v,v'$ be valuations. If $v \lu v'$, then $v' \in \erlu{v}$.
\end{lemma}
\begin{proof}
  It is not difficult to see from the definition of $\lu$
  (Definition~\ref{def:lu-preorder}) that both $v$ and $v'$ satisfy
  the same LU-guards. It remains to show the second property for $v'$
  to be in $\erlu{v}$.

  Let $x, y$ be clocks such that $\intpart{v_x} = \intpart{v'_x}$,
  $\intpart{v_y} = \intpart{v'_y}$, and $v_x \le U_x$, $v_y \le
  L_y$. Suppose $\fracpart{v_x} \lleq \fracpart{v_y}$, for $\lleq$
  being either $<$ or $=$.  As $v \lu v'$, if $v'_x > v_x$, we need
  $v_x > U_x$ which is not true. Hence we can conclude that $v'_x \le
  v_x$. Similarly, for $y$, one can conclude that $v'_y \ge v_y$. As
  the integer parts are the same in $v$ and $v'$, we get
  $\fracpart{v'_x} < \fracpart{v'_y}$ or $\fracpart{v'_x} \le
  \fracpart{v'_y}$ depending on whether $\lleq$ is $<$ or $=$.
\end{proof}

The above along with the adjustment lemma help us to reduce the
question of inclusion as a question of
intersection. \label{page:prop-incl-int}

\begin{proposition}\label{prop:incl-int}
  Let $Z,Z'$ be zones. Then, $Z \not \incl \alu(Z')$ iff there exists
  a valuation $v \in Z$ such that $\erlu{v} \cap Z'$ is empty.
\end{proposition}
\begin{proof}
  Consider the left-to-right direction. Suppose $Z \not \incl
  \alu(Z')$. Then there exists a valuation $v \in Z$ such that for
  every valuation $v' \in Z'$ we have $v \not \lu v'$. Pick an
  arbitrary $v' \in Z'$. In particular, every valuation $v_1' \in
  \nbd{v'}$ satisfies $v \not \lu v_1'$. From the Adjustment
  lemma~\ref{lem:adjustment}, we get that $v' \not \in
  \erlu{v}$. Since $v'$ is arbitrary, we get that $\erlu{v} \cap Z'$
  is empty.

  Now for the right-to-left direction. Suppose $\erlu{v} \cap Z'$ is
  empty. Then by Lemma~\ref{lem:lu-imp-rlu}, we get that $v \not \lu
  v'$ for every valuation $v' \in Z'$. This shows that $Z \not \incl
  \alu(Z')$.
\end{proof}

\subsection{Distance graphs}
\label{sec:distance-graphs}

Thanks to Proposition~\ref{prop:incl-int}, we know that to solve $Z
\not \incl \alu(Z')$, we need to check if there exists a valuation $v
\in Z$ such that its LU-region $\erlu{v}$ does not intersect with the
zone $Z'$. The focus now is to study this intersection.

We will begin with a convenient representation of zones that we use to
solve this intersection question. The standard way to represent a zone
is using a DBM. An equivalent representation is in terms of
graphs~\cite{DBLP:journals/jacm/Shostak81}
which we call here \emph{distance graphs}. 

\begin{definition}[Distance graph]\label{def:distance-graphs}
  A \emph{distance graph} $G$ has clocks as vertices, with an
  additional special vertex $x_0$ representing constant $0$. Between
  every two vertices there is an edge with a \emph{weight} of the form
  $(\lleq, c)$ where $c\in \mathbb{Z}$ and $\lleq \in \{\le, < \}$ or
  $(\lleq,c) = (<, \infty)$. An edge $x\act{\lleq c} y$ represents a
  constraint $y-x\lleq c$: or in words, the distance from $x$ to $y$
  is bounded by $c$. We let $\sem{G}$ be the set of valuations of
  clock variables satisfying all the constraints given by the edges of
  $G$ with the restriction that the value of $x_0$ is $0$. We also say
  that the \emph{$y - x$-constraint} of $\sem{G}$ is $\lleq c$ if the
  weight of the $x \xra{} y$ edge in $G$ is $\lleq c$.
\end{definition}

For readability, we will often write $0$ instead of
$x_0$. Figure~\ref{fig:zone-dist-graph} illustrates a zone with its
constraints and the corresponding distance graph. 
\begin{figure}[t!]
  \centering \input{fig/zone-dist-graph}
  \caption{An arbitrary zone with its constraints and distance graph}
  \label{fig:zone-dist-graph}
\end{figure}
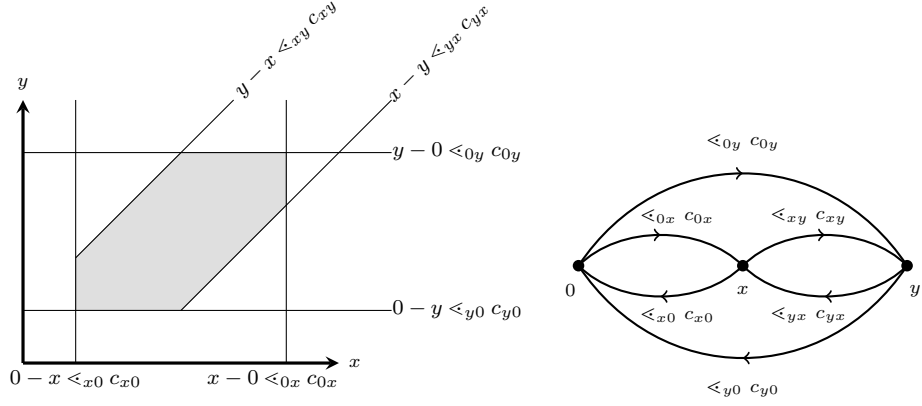

An arithmetic over the weights $(\lleq, c)$ can be defined as
follows~\cite{bengtsson2004timed}.
\begin{description}
\item \emph{Equality} $(\lleq_1,c_1) = (\lleq_2,c_2)$ if $c_1 = c_2$
  and $\lleq_1 = \lleq_2$.
\item \emph{Addition} $(\lleq_1,c_1) + (\lleq_2,c_2) = (\lleq,c_1 +
  c_2)$ where $\lleq = <$ iff either $\lleq_1$ or $\lleq_2$ is $<$.
\item \emph{Minus} $-(\lleq,c) = (\lleq, -c)$.
\item \emph{Order} $(\lleq_1,c_1) < (\lleq_2,c_2)$ if either $c_1 <
  c_2$ or ($c_1 = c_2$ and $\lleq_1 = < $ and $\lleq_2 = \le$).
\end{description}
This arithmetic lets us talk about the weight of a path as a weight of
the sum of its edges.

A cycle in a distance graph $G$ is said to be \emph{negative} if the
sum of the weights of its edges is at most $(<,0)$; otherwise the
cycle is \emph{positive}. The following proposition is folklore.
\begin{proposition}\label{prop:cycles}
  A distance graph $G$ has only positive cycles iff $\sem{G}\not=\es$.
\end{proposition}

  A distance graph is in \emph{canonical
  form} if the weight of the edge from $x$ to $y$ is the lower bound
of the weights of paths from $x$ to $y$. The canonical form only
exists for graphs without negative cycle.
Given a distance graph, its canonical form can be computed by using an
all-pairs shortest paths algorithm like
Floyd-Warshall's~\cite{bengtsson2004timed} in time $\Oo(|X|^3)$ where
$|X|$ is the number of clocks. Note that the number of vertices in the
distance graph is $|X| + 1 $. For computing the successors of a node
in the zone graph, the most complex operation is the
computation of $Z \land g$ which involves a canonicalization
operation. However, since $g$ has diagonal free constraints, the
canonicalization procedure involved to compute $Z \land g$ is easier
and costs only $\Oo(|X|^2)$~\cite{Zhao:IPL:2005}.
\medskip

Recall that we have reduced the problem $Z \not\incl \alu(Z')$ to
checking if there exists $v$ such that $\erlu{v} \cap Z'$ is
empty. For this, we need to know when the intersection of an LU-region
and a zone is empty. In the next section we will see that an LU-region
is a zone and can be represented using a distance graph. Therefore, it
boils down to asking given two distance graphs $G_1$ and $G_2$ when is
$\sem{G_1} \cap \sem{G_2}$ empty.

For two distance graphs $G_1$, $G_2$ which are not necessarily in
canonical form, we denote by $\min(G_1,G_2)$ the distance graph where
each edge has the weight equal to the minimum of the corresponding
weights in $G_1$ and $G_2$.  Even though this graph may be not in
canonical form, it should be clear that it represents intersection of
the two arguments, that is,
$\sem{\min(G_1,G_2)}=\sem{G_1}\cap\sem{G_2}$; in other words, the
valuations satisfying the constraints given by $\min(G_1,G_2)$ are
exactly those satisfying all the constraints from $G_1$ as well as
from $G_2$.

Proposition~\ref{prop:cycles} tells us that the intersection $\sem{G_1} \cap
\sem{G_2}$ is empty iff the distance graph $\min(G_1, G_2)$ has a
negative cycle.

\subsection{LU-regions as distance graphs}
\label{sec:lu-regions-as}

Our aim is to check when the intersection of $\lureg{v}$ and $Z'$ is
empty.  We saw in the previous subsection that zones can be conveniently
and canonically represented by distance graphs. Here, we will see
how we can canonically represent an $LU$-region of a valuation as a
distance graph.

We will first recall a constructive definition of Alur-Dill regions.

\medskip

\begin{definition}[Regions: constructive definition]\label{def:region-constructive}
  A region with respect to bound function $M$ is the set of valuations
  specified as follows:\label{def:regions}
  \begin{enumerate}
  \item for each clock $x\in X$, one constraint from the set:

    $\{x = c ~\mid~ c=0, \dots, M_x \} \cup \{c - 1 < x < c ~\mid~ c =
    1, \dots, M_x \} \cup \{x > M_x \}$
  
    \smallskip
  \item for each pair of clocks $x,y$ having interval constraints:
    $c-1< x <c$ and $d-1< y< d$, it is specified if $\fracpart{x}$ is
    less than, equal to or greater than $\fracpart{y}$.
  \end{enumerate}
\end{definition}

The distance graph representing a region can be constructed using the
above constructive definition of a region. 
For a valuation $v$, let $\gv$ denote the canonical distance graph
representing the region $\reg{v}$. We are now interested in getting
the LU-region of $v$, that is, $\erlu{v}$ as a distance graph. For
convenience, we will recall below the definition of LU-regions.

\medskip

\Repeat{Definition}{def:lu-region}\textbf{(LU-region)}

For a valuation $v$ we define its \emph{LU-region}, denoted
$\rrlu{v}$, to be the set of valuations $v'$ such that:
\begin{itemize}
\item $v'$ satisfies the same $LU$-guards as $v$.
 
\item For every pair of clocks $x,y$ with
  $\intpart{v(x)}=\intpart{v'(x)}$, $\intpart{v(y)} =
  \intpart{v'(y)}$, $v(x)\leq U_x$ and $v(y)\leq L_y$ we have:
  \begin{itemize}
  \item if $0< \fracpart{v(x)} < \fracpart{v(y)}$ then
    $\fracpart{v'(x)} < \fracpart{v'(y)}$.
  \item if $0 < \fracpart{v(x)} = \fracpart{v(y)}$ then
    $\fracpart{v'(x)} \leq \fracpart{v'(y)}$.
  \end{itemize}
\end{itemize}

For a valuation $v$, we need to collect all valuations $v'$ satisfying
the above two conditions to get $\lureg{v}$. We begin with a
motivating example in Figure~\ref{fig:lu-region-dist-graph}. In the
(a) part of the figure, we consider a valuation $v$ such that $v_x >
U_x$. The shaded portion in Figure~\ref{fig:lu-region-dist-graph} (a)
shows the region $\reg{v}$ and the finite valued $x-y$
and $x-0$ constraints bounding this region
(c.f. Definition~\ref{def:distance-graphs} for definition of
$x-y$-constraints and Figure~\ref{fig:zone-dist-graph} for an
illustration). The LU-region $\lureg{v}$ is shown in
Figure~\ref{fig:lu-region-dist-graph} (b). Observe that it matches
with the definition given above. But more importantly, note that it
can be seen as a transformation of $\reg{v}$ by moving its $x-y$ and
$x-0$ constraints to infinity and keeping the rest same.

We now consider a valuation $v$ with $v_y > L_y$ in
Figure~\ref{fig:lu-region-dist-graph} (c). Once again, the shaded
portion shows the region $\reg{v}$ and the relevant boundary constraints. We
depict the LU-region $\lureg{v}$ in
Figure~\ref{fig:lu-region-dist-graph} (d) matching the definition
given above. Note that it can be seen as a transformation of $\reg{v}$
by moving the $x-y$ constraint to infinity and the $0-y$ constraint up to
$L_y$. However, when we move $x-y$ to infinity, the graph that
we get would no longer be canonical. We could then consider the
canonicalization of the transformed graph.

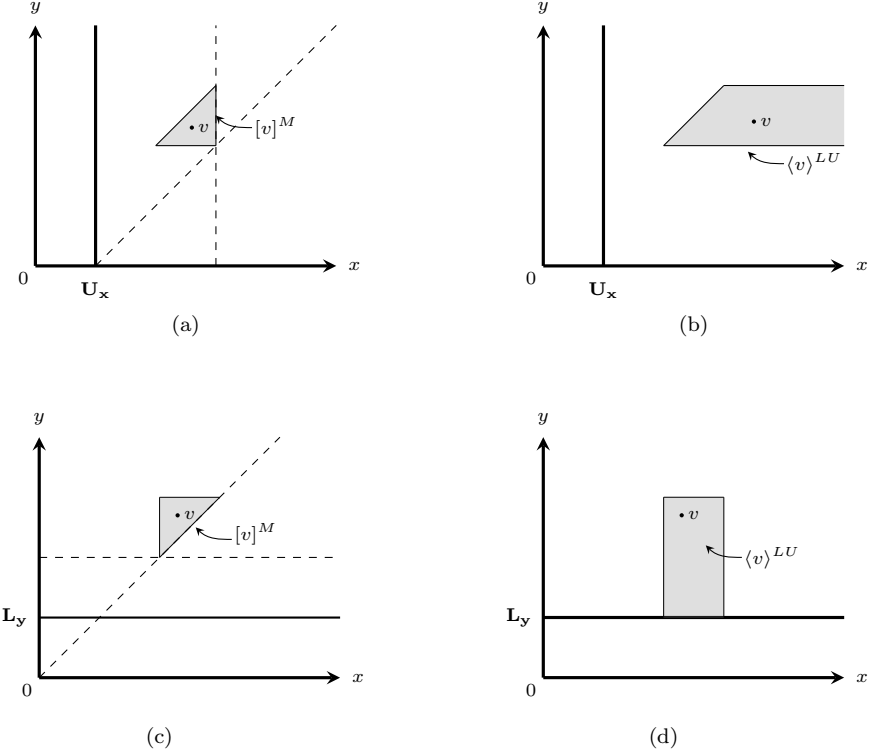
\begin{figure}[t!]
  \centering \input{fig/lureg-from-reg}
  \caption{$\lureg{v}$ can be thought of as a transformation of
    $\reg{v}$ by altering select constraints. Pictures (a) and (b)
    handle the case when $v_x >U_x$; pictures (c) and (d) handle the
    case when $v_y > L_y$}
  \label{fig:lu-region-dist-graph}
\end{figure}

These examples lead us to the definition of the distance graph $\gvlu$
for $\lureg{v}$ as a transformation of the distance graph $\gv$.

\begin{definition}[Distance graph $\gvlu$]\label{def:graph_of_LU-R}
  Let $v$ be valuation. Given the distance graph of the region
  $\reg{v}$ in canonical form $\gv = (\lleq_{xy}, c_{xy})_{x,y \in
    X}$, first define the distance graph $G' = (\lleq'_{xy},
  c'_{xy})_{x,y \in X}$ as follows:
  \begin{equation*}
    (\lleq'_{yx}, c'_{yx}) = \begin{cases} 
      (<, \infty) & \text{ if } v_x > U_x \\
      (<, \infty) & \text{ if } v_y > L_y \text{ and } x \neq
      0 \\
      (<, -L_y) & \text{ if } v_y > L_y \text{ and } x = 0 \\ 
      (\lleq_{yx}, c_{yx}) & \text{ otherwise} 
    \end{cases}
  \end{equation*}
  Then, the distance graph $\gvlu$ is defined to be the canonical form
  of $G'$.
\end{definition}

The following lemma confirms that the distance graph defined above
indeed represents $\lureg{v}$.

\begin{lemma}\label{lem:graph_of_alui}
  Let $v$ be a valuation. Let $\gvlu$ be the graph obtained by
  Definition~\ref{def:graph_of_LU-R}. Then the sets $\lureg{v}$ and
  $\sem{\gvlu}$ are equal.
\end{lemma}
\begin{proof}
  Let $v' \in \lureg{v}$.  We will now show that $v' \in
  \sem{\gvlu}$. First observe that $\sem{\gvlu} = \sem{G'}$ where $G'$
  is the graph as in Definition \ref{def:graph_of_LU-R}. Therefore it
  is sufficient to show that $v'$ satisfies the constraints given by
  $G'$. From the definition, it is clear that an edge $y \xra{} x$ is
  finite valued in $G'$ only if $v_x \le U_x$. Additionally when $v_y
  \le L_y$, the value of the edge $y \xra{} x$ is the same as that in
  $\gv$. Otherwise if $v_y > L_y$, the only finite value is $(<,-L_y)$
  for the edge $y \xra{} x_0$.

  Since $v' \in \lureg{v}$, it satisfies all the $LU$-guards that
  $v$ satisfies.
  If $y$ is a clock such that $v_y > L_y$ then $v'_y > L_y$ too.  So
  $v'$ satisfies constraints of the form $y \xra{<-L_y} x_0$. It now
  remains to look at edges $y \xra{\lleq d} x$ with $v_y \le L_y$,
  $v_x \le U_x$ and the weight $(\lleq, d)$ coming from $\gv$. Let
  $\intpart{v_x}$ and $\intpart{v_y}$ be denoted as $c_x$ and $c_y$
  respectively. As $v'$ satisfies the same $LU$-guards as $v$, we
  have:
  \begin{align}
    \label{eqn:hypothesis-vxp-vyp}
    \begin{split}
    v'_x & < c_x + 1 \\
    v'_y & \ge c_y
    \end{split}
  \end{align}
  Therefore $v'_x - v'_y < c_x + 1 - c_y$. Since $\gv$ represents the
  region containing $v$, by definition of regions, the constant in the
  weight $(\lleq, d)$ is either $c_x - c_y + 1$ or $c_x - c_y$. If it
  is the former, then clearly, $v'$ also satisfies $x - y \lleq d$.
  We need to consider the latter case, that is, $d$ is $c_x - c_y$.
  Thanks to~\eqref{eqn:hypothesis-vxp-vyp} above, if either
  $v'_x < c_x$ or $v'_y \ge c_y+1$, we are done. We are left with
  considering the case when $\intpart{v'_x} = c_x$ and
  $\intpart{v'_y} = c_y$. We have:
  \begin{align*}
    && v_x - v_y & \lleq c_x - c_y \\
    & \imp  & \fracpart{v_x} - \fracpart{v_y} & \lleq 0 \\
    & \imp & \fracpart{v'_x} - \fracpart{v'_y} & \lleq 0 \qquad
    (\text{as } v' \in \erlu{v}) \\
    & \imp &\intpart{v'_x} + \fracpart{v'_x} - (\intpart{v'_y}
    + \fracpart{v'_y}) & \lleq c_x - c_y \\
    & \imp & v'_x - v'_y & \lleq d
  \end{align*}
  This proves that if $v' \in \lureg{v}$, then $v'$ satisfies the
  constraints of $G'$ and hence $v' \in \sem{G'}$.

  Now for the other direction, assume $v' \in \sem{G'}$. We will show
  that $v' \in \lureg{v}$. Let $x,y$ be clocks such that $v_x \le U_x$
  and $v_y \le L_y$. From the definition of $\sem{G'}$, edges of the
  form $y \xra{} x_0$ and $x_0 \xra{} x$ are retained as in
  $\gv$. Since $v' \in \sem{G'}$, it is clear that $v'$ satisfies the
  same $LU$-guards as $v$. We now consider the order property for
  $LU$-regions. From Definition~\ref{def:lu-region}, the order
  property only need to be considered when
  $\intpart{v'_x} = \intpart{v_x}$ and
  $\intpart{v'_y} = \intpart{v_y}$.  Let us denote
  $c_x = \intpart{v'_x} = \intpart{v_x}$ and
  $c_y = \intpart{v'_y} = \intpart{v_y}$ By definition of $G'$, the
  edge $y \xra{} x$ in $G'$ has the same weight as that in $\gv$. Let
  the edge weight $y \xra{} x$ be $(\lleq, d)$. We have:
  \begin{align*}
    & &v_x - v_y & \lleq d \\
    &\imp & \fracpart{v_x} - \fracpart{v_y} & \lleq d - (c_x - c_y)
  \end{align*}
  If $\fracpart{v_x} < \fracpart{v_y}$ then either $d - (c_x - c_y) <
  0 $ or if it is $0$ then $\lleq$ is the strict inequality. As this
  edge remains in $G'$, the valuation $v'$ satisfies $v'_x - v'_y
  \lleq d$. Moreover, since the integral parts of $v'$ match, we get
  $\fracpart{v'_x} - \fracpart{v'_y} \lleq d - (c_x - c_y)$. By the
  aforementioned property, we get $\fracpart{v'_x} <
  \fracpart{v'_y}$. A similar argument follows for the case when
  $\fracpart{v_x} = \fracpart{v_y}$.
\end{proof}

Before we use the distance graph $\gvlu$ for further analysis, recall
that we first defined a graph $G'$ in
Definition~\ref{def:graph_of_LU-R} and then obtained $\gvlu$ by
canonicalizing it. We will now observe some properties of $\gvlu$ that
are either retained from $G'$ or obtained thanks to
canonicalization. These observations would be important in the next
section when we do the analysis on the distance graph representing
LU-region $\erlu{v}$ and zone $Z'$.

\begin{lemma}\label{lem:graph-G_v}
  Let $v$ be a valuation. Let $\gv$, $\gvlu$ be the canonical distance
  graphs of $\reg{v}$ and $\lureg{v}$ respectively. For variables
  $x,y$, if the edge $y \xra{} x$ has a finite value in $\gvlu$,
  then:
  \begin{enumerate}
  \item $v_x \le U_x$,
  \item if $v_y \le L_y$, the value of $y \xra{} x$ in $\gvlu$ and
    $\gv$ are equal,
  \item if $v_y > L_y$, the value of $y \xra{} x$ in $\gvlu$ equals
    the value of the path $y \xra{} x_0 \xra{} x$ in $\gvlu$.
  \end{enumerate}
\end{lemma}
\begin{proof}
  The graph $\gvlu$ is the canonical form of the graph $G'$ defined in
  Definition \ref{def:graph_of_LU-R}. By definition, if $v_x > U_x$,
  all incoming edges to $x$ in $G'$ have weight $(<,\infty)$. So, the
  shortest path in this graph $G'$ from a variable $y$ to a variable
  $x$ such that $v_x > U_x$ is $(<,\infty)$. Therefore, if in the
  canonical form $\gvlu$, the edge $y \xra{} x$ is finite valued, we
  should have $v_x \le U_x$. This gives the first part of the lemma.

  Consider the second part of the lemma. We know that $v_y \le L_y$
  and from the first part of the lemma, we know that $v_x \le U_x$.
  The weight of $y \xra{} x$ in $G'$ is the same as that of $\gv$
  according to Definition \ref{def:graph_of_LU-R}. Note that the
  finite values in the graph $G'$ are either the same as that of $\gv$
  or of the form $(<, -L_z)$ for some edges $z \xra{} 0$. In the
  latter case, we also know by definition that $v_z > L_z$.  Therefore
  the value $(<, -L_z)$ is greater than the corresponding value in
  $\gv$. As $\gv$ is canonical, the shortest path from $y$ to $x$ in
  $G'$ cannot reduce from its value in $\gv$ and hence equals just the
  edge value $y \xra{} x$. This gives the second part of the lemma.

  Finally consider the third part. Assume $v_y > L_y$. From Part 1, we
  know that $v_x \le U_x$. By Definition \ref{def:graph_of_LU-R}, the
  weight of $y \xra{} x$ equals $(<,\infty)$ if $x$ is not $x_0$. The
  only finite valued outgoing edge from $y$ is $y \xra{}
  x_0$. Therefore, we can infer two things: the shortest path from $y$
  to $x_0$ is given by the edge $y \xra{} x_0$; and the shortest path
  from $y$ to $x$ should contain this edge $y \xra{} x_0$. Secondly,
  note that variable $x_0$ has $v_{x_0} \le L_{x_0}$ ($v_{x_0} = 0 =
  L_{x_0}$). By definition, the value of $x_0 \xra{} x$ in $G'$ is
  given by the corresponding value in $\gv$ and by Part 2, we know
  that this value stays in $\gvlu$, that is, the shortest path from
  $x_0$ to $x$ in $G'$ is given by the direct edge. Summing up, the
  shortest path from $y$ to $x$ in $G'$ is given by $y \xra{} x_0
  \xra{} x$, where both $y\xra{} x_0$ and $x_0 \xra{} x$ are values
  coming from $\gvlu$.
\end{proof}

\subsection{When does an $LU$-region intersect a zone.}
\label{sec:when-does-an}

We are now in a position to characterize the intersection $\erlu{v}
\cap Z'$. Let $\gvlu$ as defined in the previous section be the
canonical distance graph of $\erlu{v}$ and let $G_{Z'}$ be the
canonical distance graph of $Z'$. By Proposition \ref{prop:cycles},
the intersection $\erlu{v} \cap Z'$ is empty iff $\min(\gvlu, G_{Z'})$
has a negative cycle.

We will now state a necessary and sufficient condition for the distance graph
$\min(\gvlu, G_{Z'})$ to have a negative cycle. We denote by $Z'_{xy}$
the weight of the edge $x \xra{} y$ in $G_{Z'}$. Similarly we denote
$\lureg{v}_{xy}$ for the weight of $x \xra{} y$ in $\gvlu$. When a
variable $x$ represents the special clock $x_0$, we define
$\lureg{v}_{0x}$ and $\lureg{v}_{x0}$ to be $(\le, 0)$. Since by
convention $x_0$ is always $0$, this is consistent.
We also denote $\reg{v}_{xy}$ for the weight of $x \xra{} y$ in
$\gv$ with the same convention when $x=x_0$.

The next proposition is the most important observation used in getting
the final inclusion test.
\begin{proposition}\label{prop:neg-cycles-in-min}
  Let $v$ be a valuation and $Z'$ a zone.  The intersection $\lureg{v}
  \cap Z'$ is empty iff there exist two variables $x,y$ such that $v_x
  \le U_x$ and $Z'_{xy} + \lureg{v}_{yx} < (\le,0)$.
\end{proposition}

To prove the above proposition, we need a small but a crucial
observation that exploits the special structure of regions. A variable
$x$ is said to be bounded in valuation $v$ if $v_x \le \max(L_x,
U_x)$. If $x,y$ are bounded in $v$, then the projection of the region
$\reg{v}$ onto $x,y$ has very specific boundaries. Call it a bounded
region. The following lemma makes use of the fact that a bounded
region is either fully contained in a zone or is totally disjoint from
it, that is, there cannot be a partial intersection of the bounded
region and zone, as illustrated in Figure~\ref{fig:bounded-regions}.

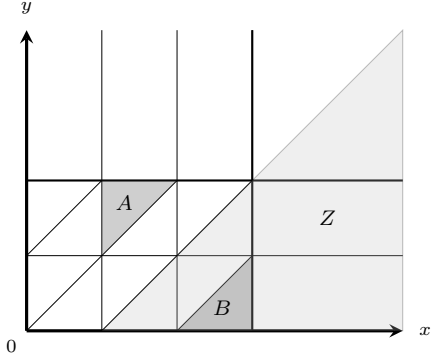
\begin{figure}
  \centering \input{fig/bounded-regions}
  \caption{A bounded region is either totally contained in a zone $Z$
    or totally disjoint from it. It can never intersect partially. For
    example, the bounded region B is totally inside $Z$, however the
    bounded region $A$ is totally disjoint from $Z$.}
  \label{fig:bounded-regions}
\end{figure}

\begin{lemma}\label{lem:bounded-vars}
  Let $x$, $y$ be bounded variables of $v$ appearing in some negative
  cycle $N$ of $\min(\gvlu, G_{Z'})$.  Let the edge weights be $x
  \xra{\lleq_{xy} c_{xy}} y$ and $y \xra{\lleq_{yx} c_{yx}} x$ in
  $\gv$.  If the value of the path $x \xra{} \dots \xra{} y$ in $N$ is
  strictly less than $(\lleq_{xy},c_{xy})$, then $x \xra{} \dots
  \xra{} y \xra{\lleq_{yx} c_{yx}} x$ is a negative cycle.
\end{lemma}
\begin{proof}
  Let the path $x \xra{} \dots \xra{} y$ in $N$ have weight $(\lleq,
  c)$. Now, since $x$ and $y$ are bounded variables in $v$, we can
  have either $y - x = d$ or $d-1 < y-x < d$ for some integer $d$ in
  $\gv$.

  In the first case, we have edges $x \xra{\le d} y$ and $y \xra{\le
    -d} x$ in $\gv$, that is $(\lleq_{xy}, c_{xy}) = (\le, d)$ and
  $(\lleq_{yx}, c_{yx}) = (\le, -d)$. Since by hypothesis $(\lleq,c)$
  is strictly less than $(\le, d)$, we have either $c < d$ or $c = d$
  and $\lleq$ is the strict inequality. Hence $(\lleq,c) + (\le, -d) <
  (\le, 0)$ showing that $x \xra{} \dots \xra{} y \xra{\lleq_{yx}
    c_{yx}} x$ is a negative cycle.

  In the second case, we have edges $x \xra{< d} y$ and $y \xra{ < -d
    + 1} x$ in $\gv$, that is, $(\lleq_{xy}, c_{xy}) = (<,d)$ and
  $(\lleq_{yx}, c_{yx}) = (<,-d)$. Here $c < d$ and again $x \xra{}
  \dots \xra{} y \xra{\lleq_{yx} c_{yx}} x$ gives a negative cycle.
\end{proof}

We can now prove Proposition~\ref{prop:neg-cycles-in-min}.

\subsubsection*{Proof of Proposition \ref{prop:neg-cycles-in-min}}
The distance graph $\min(\gvlu, G_{Z'})$ represents the set $\lureg{v}
\cap Z'$. By Proposition \ref{prop:cycles}, the intersection is empty
iff $\min(\gvlu, G_{Z'})$ has a negative cycle. If there exist
variables $x,y$ such that $Z'_{xy} + \lureg{v}_{yx} < (\le, 0)$, then
there is a negative cycle $x \xra{} y \xra{} x$ in $\min(\gvlu,
G_{Z'})$ and hence $\lureg{v} \cap Z'$ is empty. This shows the
right-to-left direction.

The left-to-right direction is less trivial. Assume that $\lureg{v}
\cap Z'$ is empty. Then, there is a negative cycle $N$ in $\min(\gvlu,
G_{Z'})$. To prove the proposition, we aim to show the following.

\paragraph{Aim} We show that the negative cycle $N$ of $\min(\gvlu,
G_{Z'})$ can be reduced to the form:
\begin{equation}
  \label{eq:aim}
  x \xra{~Z'_{xy}~} y  \xra{~\lureg{v}_{yx}~} x
\end{equation}

Firstly, since both $\gvlu$ and $G_{Z'}$ are canonical, we can assume
without loss of generality that no two consecutive edges in $N$ come
from the same graph.

Suppose there are two edges $y_1 \xra{} x_1$ and $y_2 \xra{} x_2$ in
$N$ with weights coming from $\gvlu$:
\begin{align}
  \label{eq:two-edges}
  y_1 \xra{~~\lureg{v}_{y_1x_1}~~} x_1 \xra{} \cdots \xra{} y_2
  \xra{~~\lureg{v}_{y_2x_2}~~} x_2 \xra{} \cdots \xra{} y_1
\end{align}
Since they are part of a negative cycle, their edge weights should be
a finite value and by Part 1 of Lemma \ref{lem:graph-G_v}, this means:
\begin{align*}
  v_{x_1} \le U_{x_1} \text{ and } v_{x_2} \le U_{x_2}
\end{align*}

\paragraph{1. Suppose $v_{y_1} \le L_{y_1}$ and $v_{y_2} \le L_{y_2}$}
By Part 2 of Lemma \ref{lem:graph-G_v}, the edge values $y_1 \xra{}
x_1$ and $y_2 \xra{} x_2$ are the same as in $\gv$.  Consider the
edge:
\begin{align*}
  \text{ $y_1 \xra{} x_2$ in $\gvlu$ }
\end{align*}
Again, from the same lemma, this edge value comes from $\gv$ too.

If the value of this edge $y_1 \xra{} x_2$ is smaller than the value
of the path $y_1 \xra{} x_1 \xra{} \cdots \xra{} y_2 \xra{} x_2$ in
$N$, then this path can be replaced by the single edge $y_1 \xra{}
x_2$ to get a smaller negative cycle in $\min(\gvlu, G_Z)$.

However, if the value of the path $y_1 \xra{} x_1 \xra{} \cdots \xra{}
y_2 \xra{} x_2$ is less than the edge value $y_1 \xra{} x_2$, then by
Lemma \ref{lem:bounded-vars}:
\begin{align*}
  \text{ $y_1 \xra{} x_1 \xra{} \dots \xra{} y_2
    \xra{} x_2 \xra{}
    y_1$, where $x_2 \xra{} y_1$ comes from $\gv$ }
\end{align*}
is a negative cycle. The edge $x_2 \xra{} y_1$ might be infinity in
$\gvlu$. But as $\gv$ is canonical, we can replace $y_2 \xra{} x_2
\xra{} y_1 \xra{} x_1$ by $y_2 \xra{} x_1$. From Lemma
\ref{lem:graph-G_v}, this edge is retained in $\gvlu$ and hence we get
a smaller negative cycle.

Therefore in this case, we can eliminate the two edges $y_1 \xra{}
x_1$ and $y_2 \xra{} x_2$ to get a smaller negative cycle containing
either $y_1 \xra{} x_2$ or $y_2 \xra{} x_1$. If $N$ does not contain a
variable $z$ such that $v_z > L_z$, this elimination can be repeatedly
applied and $N$ can be reduced to a negative cycle of the form $y
\xra{} x \xra{} y$ with $v_y \le L_y$, $v_x \le U_x$ and the edge
weights $y \xra{} x$ coming from $\gvlu$ and $x \xra{} y$ coming from
$G_{Z'}$, exactly as required by (\ref{eq:aim}).

\paragraph{2. Suppose $v_{y_1} > L_{y_1}$}
Consider again the two edges $y_1 \xra{} x_1$ and $y_2 \xra{} x_2$ of
(\ref{eq:two-edges}) and now suppose that $v_{y_1} > L_{y_1}$. By Part
3 of Lemma \ref{lem:graph-G_v}, the edge $y_1 \xra{} x_1$ can be
replaced by:
\begin{align*}
  \text{$y_1 \xra{} x_0 \xra{} x_1$ of $\gvlu$}
\end{align*}
If there is another variable in $N$ that is greater than its $L$
bound, then the vertex $x_0$ would occur twice in the negative
cycle. From this negative cycle, we can obtain a smaller negative
cycle containing only one occurrence of $x_0$. Hence, without loss of
generality, we can assume that $x_0$ occurs only once in $N$. In
particular, this gives us that:
\begin{align*}
  \text{$v_{y_2} \le L_{y_2}$}
\end{align*}

Note that the special variable $x_0$ has $v_{x_0} \le L_{x_0}$ as its
value is always supposed to be $0$ and $L_{x_0}$ is defined to be
$0$. Now consider the two edges:
\begin{align*}
  x_0 \xra{} x_1 \text{ and } y_2 \xra{} x_2
\end{align*}
This corresponds to Case 1 as $v_{x_0} \le L_{x_0}$ and $v_{y_2} \le
L_{y_2}$. As we have seen, these two edges can be eliminated to give a
smaller negative cycle containing either $x_0 \xra{} x_2$ or $y_2
\xra{} x_1$, with the respective value coming from $\gvlu$.
 
If it is the latter edge $y_2 \xra{} x_1$, the smaller negative cycle
does not contain $y_1$ and hence all variables are bounded by $L$. By
Case 1, it can be reduced to a cycle as required by the proposition.

Let us now consider the former edge $x_0 \xra{} x_2$. We have the
cycle:
\begin{align*}
  \text{$y_1 \xra{} x_0 \xra{} x_2 \xra{} \dots \xra{} y_1$}
\end{align*}
All the variables other than $y_1$ in the path $x_0 \xra{} \dots
\xra{} y_1$ are bounded by their $L$ bound. We can therefore assume
that all edges in $x_2 \xra{} \dots \xra{} y_1$ come from $G_{Z'}$,
because if not, we can apply the argument of Case 1 to further reduce
the cycle. As $G_{Z'}$ and $\gvlu$ are canonical, this cycle reduces
to $y_1 \xra{} x_2 \xra{} y_1$ with $y_1 \xra{} x_2$ coming from
$\gvlu$ and $x_2 \xra{} y_1$ coming from $Z'$. This again conforms to
the form of the cycle required by (\ref{eq:aim}).  \koniec

\subsection{Final steps}
\label{sec:final-step}

Proposition~\ref{prop:neg-cycles-in-min} gives a useful
characterization of when $\lureg{v} \cap Z'$ is empty. To lift this
characterization to $Z \not \incl \alu(Z')$ and
Proposition~\ref{prop:incl-int}, we need to find the least value of
$\lureg{v}_{yx}$ from among the valuations $v \in Z$ and see if this
satisfies the condition given in
Proposition~\ref{prop:neg-cycles-in-min}.

For the moment, assume that $L=U$ so that the LU-regions coincide with
the classic regions. Consider a zone $Z$ lying within the $M$ bounds,
as shown in Figure~\ref{fig:final-step}. The values of $\reg{v}_{yx}$
for different valuations in the zone are shown. The value decreases as
we move towards the ``left boundary''. In the figure, since the
$y - x$ constraint of $Z$ is given by $y - x \le -1$, there exists a
valuation $v_4$ on the ``boundary'' and hence the least value of
$\reg{v}_{yx}$ among all $v \in Z$ would be given by $\reg{v_4}_{yx}$
which is $(\le, 1)$. If the $y -x $ constraint of $Z$ was $y - x < -1$
with a strict inequality, then the least value of $\reg{v}_{yx}$ from
$v \in Z$ is no longer $(\le, 1)$ as there is no valuation attaining
this value. In this case, the least value of $\reg{v}_{yx}$ would be
$(<, 2)$, given by the open region containing $v_5$ with constraints
$y - x <-1$ and $x - y < 2$.

\begin{figure}[!t]
  \centering \input{fig/min-lureg}
  \caption{Value of $\reg{v}_{yx}$ decreases as we move left}
  \label{fig:final-step}
\end{figure}
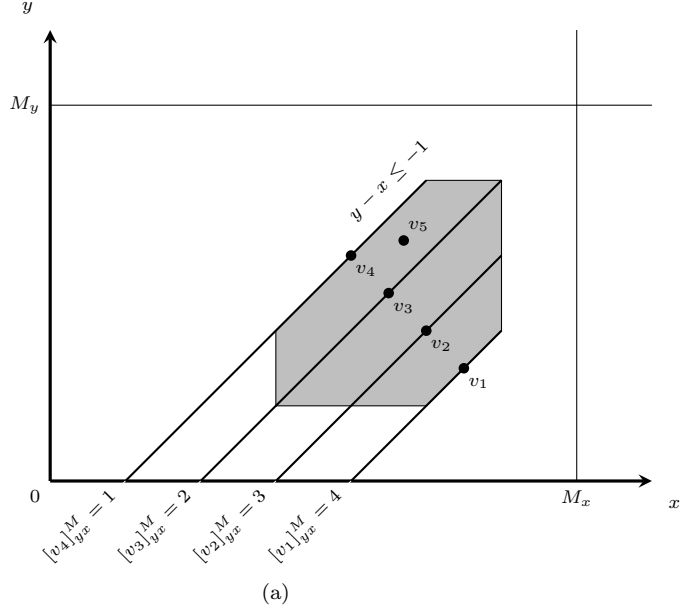

Due to this asymmetry, we need to define the following notion to
handle weights in a convenient way. For a weight $(\lleq, c)$ we
define $-(\lleq, c)$ as $(\lleq, -c)$ and a \emph{ceiling}
function $\ceil{\cdot}$ as follows.

\begin{definition}
  For a real $c$, let $\ceil{c}$ denote the smallest integer that is
  greater than or equal to $c$. We define the \emph{ceiling} function
  $\ceil{(\lleq, c)}$ for a weight $(\lleq, c)$ depending on whether
  $\lleq$ equals $\le$ or $<$, as follows:
 
\begin{equation*}
  \ceil{(\le,c)} = \begin{cases} (\le, c) & \text{if } c \text{ is an
      integer} \\
    (<, \ceil{c}) & \text{otherwise } \end{cases}
\end{equation*}

\begin{equation*}
  \ceil{(<,c)} = \begin{cases} (<, c+1) & \text{if } c \text{ is an
      integer} \\
    (<, \ceil{c}) & \text{otherwise } \end{cases}
\end{equation*}

\end{definition}

The following proposition is one of the two cores for the proof of the main
theorem. It gives the least value of $\lureg{v}_{yx}$ from among the
valuations $v$ present in a zone $Z$.

\begin{proposition}\label{prop:minimum-val-of-xy}
  Let $Z$ be a non-empty zone, and let $x$, $y$ be two clocks. The least
  value of $\lureg{v}_{yx}$ among the valuations $v \in Z$ such that
  $v_x\leq U_x$ is given by:
  \begin{equation*}
    \begin{cases}
      (<, \infty) & \text{ if } Z_{x0} < (\le, -U_x) \\
      \max\{\ceil{-Z_{xy}}, \ceil{-Z_{x0}} - (<, L_y) \} & \text{
        otherwise}
    \end{cases}
  \end{equation*}  
\end{proposition}
\begin{proof}
  Let $G$ be the canonical distance graph representing the zone
  $Z$. We denote the weight of an edge $i \xra{} j$ in $G$ by
  $(\lleq_{ij}, c_{ij})$. Recall that this means $Z_{ij} =
  (\lleq_{ij}, c_{ij})$.

  We are interested in computing the smallest value of the $x-y$
  constraint defining an $LU$-region intersecting $Z$. Additionally we
  want to restrict to $LU$-regions in which all its valuations satisfy
  $x \le U_x$, that is, we need to find:
  \begin{equation*}
    \b:= \min\{\lureg{v}_{yx} ~|~ v \in Z~\text{and}~v_x\le U_x \}
  \end{equation*}
  Clearly, if $v_x > U_x$ for all valuations $v \in
  Z$, then $\b$ is $(<,\infty)$. When $Z_{x0} < (\le, -U_x)$, it means
  that all valuations $v \in Z$ satisfy $0 - v_x \lleq_{x0} c_{x0}$
  and $c_{x0} \le -U_x$. Moreover $\lleq_{x0}$ is the strict
  inequality if $c_{x0} = -U_x$. In consequence, all valuations $v \in
  Z$ satisfy $v_x > U_x$ when $Z_{x0} < (\le, - U_x)$. Whence $\b =
  (<, \infty)$. This corresponds to the first case in the statement of
  the lemma.

  Let now restrict to the case when $Z_{x0} \ge (\le, -U_x)$.  By
  definition of regions (cf. Definition \ref{def:regions}) and Lemma
  \ref{lem:graph-G_v}, we have for a valuation $v$:
  \begin{equation}\label{xy-value}
    \lureg{v}_{yx} = \begin{cases}
      \ceil{(\le, v_x - v_y)} & \text{if } v_x \le U_x \text{ and } v_y \le L_y \\
      (<,-L_y) + \ceil{(\le, v_x)} & \text{if } v_x \le U_x \text{
        and } v_y> L_y
    \end{cases}
  \end{equation}

  Let $G'$ be the graph in which the edge $0 \xra{} x$ has weight
  $\min\{ (\le, U_x), (\lleq_{0x}, c_{0x})\}$ and the rest of the
  edges are the same as that of $G$. This graph $G'$ represents the
  valuations of $Z$ that have $v_x \le U_x$:
  $\sem{G'} = \{ v \in Z ~|~ v_x\le U_x\}$. We show that this set is
  not empty. For this we check that $G'$ does not have negative
  cycles. Since $G$ does not have negative cycles, every negative
  cycle in $G'$ should include the newly modified edge $0 \xra{} x$.
  Note that the shortest path value from $x$ to $0$ does not change
  due to this modified edge. So the only possible negative cycle in
  $G'$ is $0 \xra{} x \xra{} 0$. But then we are considering the case
  when $Z_{x0} \ge (\le, -U_x)$, and so
  $Z_{x0} + (\le, U_x) \ge (\le,0)$. Hence this cycle cannot be
  negative either. In consequence all the cycles in $G'$ are positive
  and $\sem{G'}$ is not empty.

  To find $\b$, it is sufficient to consider only the valuations in
  $\sem{G'}$. As seen from Equation \ref{xy-value}, among the
  valuations in $\sem{G'}$, we need to differentiate between those
  with $v_y \le L_y$ and the ones with $v_y > L_y$.  We proceed as
  follows. We first compute $\min\{ \lureg{v}_{yx}~|~v \in \sem{G'}
  \text{ and } v_y \le L_y\}$. Call this $\b_1$. Next, we compute
  $\min\{[v]_{yx}~|~v \in \sem{G'} \text{ and } v_y > L_y\}$ and set
  this as $\b_2$. Our required value $\b$ would then equal
  $\min\{\b_1, \b_2\}$.

  To compute $\b_1$, consider the following distance graph $G'_1$
  which is obtained from $G'$ by just changing the edge $0 \xra{} y$
  to $\min\{ (\le, L_y), (\lleq_{0y}, c_{0y}) \}$ and keeping the
  remaining edges the same as in $G'$. The set of valuations
  $\sem{G'_1}$ equals $\{ v \in \sem{G'} ~|~ v_y \le L_y \}$.
  If $\sem{G'_1} = \es$, we set $\b_1$ to $(<, \infty)$ and proceed to
  calculate $\b_2$. If not, we see that from Equation \ref{xy-value},
  for every $v \in \sem{G'_1}$, $[v]_{yx}$ is given by $\ceil{(\le,
    v_x - v_y)}$. Let $(\lleq_1,w_1)$ be the shortest path from $x$ to
  $y$ in the graph $G'_1$. Then, we have for all $v \in \sem{G'_1}$,
  $v_y - v_x \lleq_1 w_1$. If $\lleq_1$ is $\le$, then the least value
  of $[v]_{yx}$ would be $(\le, -w_1)$ and if $\lleq_1$ is $<$, one
  can see that the least value of $[v]_{yx}$ is $(<, -w_1+1)$. This
  shows that $\b_1 = \ceil{(\lleq_1, -w_1)}$. It now remains to
  calculate $(\lleq_1, w_1)$.

  Recall that $G'_1$ has the same edges as in $G$ except possibly
  different edges $0 \xra{} x$ and $0 \xra{} y$. If the shortest path
  from $x$ to $y$ has changed in $G'_1$, then clearly it should be due
  to one of the above two edges. However note that the edge $0 \xra{}
  x$ cannot belong to the shortest path from $x$ to $y$ since it would
  contain a cycle $x \xra{} \dots 0 \xra{} x \xra{} \dots y$ that can
  be removed to give shorter path. Therefore, only the edge $0 \xra{}
  y$ can potentially yield a shorter path: $x \xra{} \dots \xra{} 0
  \xra{} y$. However, the shortest path from $x$ to $0$ in $G'_1$
  cannot change due to the added edges since that would form a cycle
  with $0$ and we know that all cycles in $G'_1$ are positive.
  Therefore the shortest path from $x$ to $0$ is the direct edge $x
  \xra{} 0$, and the shortest path from $x$ to $y$ is the minimum of
  the direct edge $x \xra{} y$ and the path $x \xra{} 0 \xra{} y$. We
  get: $(\lleq_1, w_1) = \min\{ (\lleq_{xy}, c_{xy}), (\lleq_{x0},
  c_{x0}) + (\le, L_y) \}$ which equals $\min\{ Z_{xy}, Z_{x0} + (\le,
  L_y)\}$. Finally, from the argument in the above two paragraphs, we
  get:

\begin{equation}\label{b1}
  \b_1 = \begin{cases}
    (<, \infty) & \text{ if } \sem{G'_1} = \es \\
    \ceil{-Z_{xy}} & \text{ if } \sem{G'_1} \neq \es \text{ and } Z_{xy}
    \le Z_{x0} + (\le, L_y) \\
    \ceil{-Z_{x0}} + (\le, -L_y) & \text{ if } \sem{G'_1} \neq \es
    \text{ and } Z_{xy} > Z_{x0} + (\le, L_y) \\
  \end{cases}
\end{equation}

We now proceed to compute $\b_2=\min\{ [v]_{yx} ~|~ v \in \sem{G'}
\text{ and } v_y > L_y \}$. Let $G'_2$ be the graph which is obtained
from $G'$ by modifying the edge $y \xra{} 0$ to $\min\{ Z_{y0},
(<,-L_y)\}$ and keeping the rest of the edges the same as in
$G'$. Clearly $\sem{G'_2} = \min\{ v \in \sem{G'}~|~ v_y >
L_y\}$.

Again, if $\sem{G'_2}$ is empty, we set $\b_2$ to
$(<,\infty)$. Otherwise, from Equation \ref{xy-value}, for each
valuation $v \in \sem{G'_2}$, the value of $[v]_{yx}$ is given by
$(<,\ceil{v_x} - L_y)$. For the minimum value, we need the least value
of $v_x$ from $v \in \sem{G'_2}$. Let $(\lleq_2, w_2)$ be the shortest
path from $x$ to $0$ in $G'_2$. Then, since $-v_x \lleq_2 w_2$, the
least value of $\ceil{v_x}$ would be $-w_2$ if $\lleq_2$ is $\le$ and
equal to $\ceil{-w_2}$ if $\lleq_2 = <$ and $\b_2$ would respectively
be $(<, -w_2 -L_y)$ or $(<, -w_2 + 1 -L_y)$. It now remains to
calculate $(\lleq_2, w_2)$.

Recall that $G'_2$ is $G$ with $0 \xra{} x$ and $y \xra{} 0$
modified. The shortest path from $x$ to $0$ cannot include the edge $0
\xra{} x$ since it would need to contain a cycle, for the same reasons
as in the $\b_1$ case.  So we get $(\lleq_2, w_2) = \min\{ Z_{x0},
Z_{xy} + (<,-L_y)\}$. If $Z_{x0} \le Z_{xy} + (<,-L_y)$, then we take
$(\lleq_2, w_2)$ as $Z_{x0}$, otherwise we take it to be $Z_{xy} + (<,
-L_y)$. So, we get $\b_2$ as the following:

\begin{equation}\label{b2}
  \b_2 = \begin{cases}
    (<, \infty) & \text{ if } \sem{G'_2} = \es \\
    -Z_{xy} + (<,1) & \text{ if } \sem{G'_2} \neq \es \text{ and } Z_{x0} \ge Z_{xy} + (<,-L_y) \\
    \ceil{-Z_{x0}} + (<, -L_y) & \text{ if } \sem{G'_2} \neq \es
    \text{ and } Z_{x0} < Z_{xy} + (<,
    -L_y) \\
  \end{cases}
\end{equation}
However, we would like to write $\b_2$ in terms of the cases used for
$\b_1$ in Equation \ref{b1} so that we can write $\b$, which equals
$\min\{ \b_1, \b_2\}$, conveniently.

Let $\p_1$ be the inequation: $Z_{xy} \le Z_{x0} + (\le, L_y)$. From
Equation \ref{b1}, note that $\b_1$ has been classified according to
$\p_1$ and $\neg \p_1$ when $\sem{G'_1}$ is not empty. Similarly, let
$\p_2$ be the inequation: $Z_{x0} \ge Z_{xy} + (<,-L_y)$. From
Equation \ref{b2} we see that $\b_2$ has been classified in terms of
$\p_2$ and $\neg \p_2$ when $\sem{G'_2}$ is not empty. Notice the
subtle difference between $\p_1$ and $\p_2$ in the weight component
involving $L_y$: in the former the inequality associated with $L_y$ is
$\le$ and in the latter it is $<$. This necessitates a bit more of
analysis before we can write $\b_2$ in terms of $\p_1$ and $\neg
\p_1$.

Suppose $\p_1$ is true. So we have $(\lleq_{xy},c_{xy}) \le
(\lleq_{x0}, c_{x0} + L_y)$. This implies: $c_{xy} \le c_{x0} +
L_y$. Therefore, $c_{x0} \ge c_{xy} - L_y$. When $c_{x0} > c_{xy} -
L_y$, $\p_2$ is clearly true. For the case when $c_{x0}= c_{xy} -
L_y$, note that in $\p_2$ the right hand side is always of the form
$(<, c_{xy} - L_y)$, irrespective of the inequality in $Z_{xy}$ and so
yet again, $\p_2$ is true. We have thus shown that $\p_1$ implies
$\p_2$. 

Suppose $\neg \p_1$ is true.  We have $(\lleq_{xy}, c_{xy}) >
(\lleq_{x0}, c_{x0} + L_y)$. If $c_{xy} > c_{x0} + L_y$, then clearly
$c_{x0} < c_{xy} - L_y$ implying that $\neg \p_2$ holds. If $c_{xy} =
c_{x0} + L_y$, then we need to have $\lleq_{xy}$ equal to $\le$ and
$\lleq_{x0}$ equal to $<$. Although $\neg \p_2$ does not hold now, we
can safely take $\b_2$ to be $\ceil{-Z_{x0}} + (<, -L_y)$ as its value
is in fact equal to $-Z_{xy} + (<,1)$ in this case. Summarizing the
above two paragraphs, we can rewrite $\b_2$ as follows:

\begin{equation}\label{b2-final}
  \b_2 = \begin{cases}
    (<, \infty) & \text{ if } \sem{G'_2} = \es \\
    -Z_{xy} + (<,1) & \text{ if } \sem{G'_2} \neq \es \text{ and } Z_{xy} \le Z_{x0} + (\le,L_y) \\
    \ceil{-Z_{x0}} + (<, -L_y) & \text{ if } \sem{G'_2} \neq \es
    \text{ and } Z_{xy} > Z_{x0} + (\le, L_y) \\
  \end{cases}
\end{equation}

We are now in a position to determine $\b$ as $\min\{\b_1, \b_2
\}$. Recall that we are in the case where $Z_{x0} \le (\le, -U_x)$ and
we have established that $\sem{G'}$ is non-empty. Now since $\sem{G'}
= \sem{G'_1} \cup \sem{G'_2}$ by construction, both of them cannot be
simultaneously empty. Hence from Equations \ref{b1} and
\ref{b2-final}, we get $\b$, the $\min\{\b_1, \b_2 \}$ as:

\begin{equation}\label{b}
  \b_ = \begin{cases}
    \ceil{-Z_{xy}} & \text{ if } Z_{xy} \le Z_{x0} + (\le,L_y) \\
    \ceil{-Z_{x0}} + (<, -L_y) & \text{ if } Z_{xy} > Z_{x0} + (\le, L_y) \\
  \end{cases}
\end{equation}

There remains one last reasoning. To prove the lemma, we need to show
that $\b = \max\{ \ceil{-Z_{xy}}, \ceil{-Z_{x0}} + (<,-L_y)\}$.  For
this it is enough to show the following two implications:

\begin{align*}
  Z_{xy} \le Z_{x0} + (\le, L_y) \imp \ceil{-Z_{xy}} \ge
  \ceil{-Z_{x0}} + (<, -L_y) \\
  Z_{xy} > Z_{x0} + (\le, L_y) \imp \ceil{-Z_{xy}} \le
  \ceil{-Z_{x0}}+(<,-L_y)
\end{align*}
We prove only the first implication. The second follows in a similar
fashion.  Let us consider the notation $(\lleq_{xy}, c_{xy})$ and
$(\lleq_{x0}, c_{x0})$ for $Z_{xy}$ and $Z_{x0}$ respectively.  So we
have:

\begin{align*}
  & (\lleq_{xy}, c_{xy}) \le (\lleq_{x0},c_{x0}) + (\le, L_y) \\
  \imp~ & (\lleq_{xy}, c_{xy}) \le (\lleq_{x0}, c_{x0} + L_y)
\end{align*}
If the constant $c_{xy} < c_{x0} + L_y$, then $-c_{xy} > -c_{x0} -
L_y$ and we clearly get that $\ceil{-Z_{xy}} \ge \ceil{-Z_{x0}} +
(<,-L_y)$. If the constant $c_{xy} = c_{x0} + L_y$ and if $\lleq_{x0}$
is $\le$, then the required inequation is trivially true; if
$\lleq_{x0}$ is $<$, it implies that $\lleq_{xy}$ is $<$ too and
clearly $\ceil{(<,-c_{xy})}$ equals $ \ceil{(<,-c_{x0})} + (<,-L_y)$.
\end{proof}

\medskip

We have a simple method that tells us when an LU-region $\lureg{v}$
does not intersect a zone $Z'$
(Proposition~\ref{prop:neg-cycles-in-min}). We have also characterized
the potential valuation $v$ from $Z$ that could satisfy the
non-intersection condition with $Z'$
(Proposition~\ref{prop:minimum-val-of-xy}). This gives the necessary
tools to solve inclusion $Z \not \incl \alu(Z')$.  The following
theorem presents the efficient inclusion test.

\begin{theorem}\label{thm:effic-inc-test}
  Let $Z, Z'$ be non-empty zones. Then, $Z \not \incl \alu(Z')$ iff
  there exist two variables $x,y$ such that:
  \begin{equation*}
    Z_{x0} \ge (\le,-U_x)~\text{and}~Z'_{xy} <
    Z_{xy}~\text{and}~Z'_{xy} + (<,-Ly) <  Z_{x0}
  \end{equation*}
\end{theorem}
\begin{proof}
  From Proposition~\ref{prop:incl-int}, we know that $Z \not \incl
  \alu(Z')$ iff there exists a valuation $v \in Z$ such that
  $\lureg{v}$ does not intersect $Z'$.

  From Proposition \ref{prop:neg-cycles-in-min}, we know that
  $\lureg{v} \cap Z'$ is empty iff there exists a variable $x$ such
  that $v_x \le U_x$ and a variable $y$ such that:
  \begin{equation}\label{eq:recall-prop}
    Z'_{xy} + \lureg{v}_{yx} < (\le,0)
  \end{equation}
  This is possible for variables $x,y$ iff the least value of
  $\lureg{v}_{yx}$ from among the valuations in $Z$ satisfies the
  inequation~(\ref{eq:recall-prop}) with $Z'_{xy}$.

  This is where we use Proposition
  \ref{prop:minimum-val-of-xy}. According to this proposition, for
  (\ref{eq:recall-prop}) to be true for some valuation $v \in Z$, we
  would need $Z_{x0} \ge (\le, -U_x)$ and:
  \begin{equation}\label{eq:thm-1}
    Z'_{xy} + \ceil{-Z_{xy}} < (\le,0) ~\text{and}~Z'_{xy} +
    \ceil{-Z_{x0}} -(<,L_y) < (\le,0)
  \end{equation}

  Consider the first inequality: $Z'_{xy} + \ceil{-Z_{xy}} <
  (\le,0)$. Let $Z_{xy}$ be $(\lleq_{xy}, c_{xy})$. If $\lleq_{xy}$ is
  the weak inequality $\le$, then $\ceil{-Z_{xy}}$ is $(\le, -c_{xy})$
  and hence the condition becomes: $Z'_{xy} + (\le, -c_{xy}) < (\le,
  0)$. This is equivalent to saying $Z'_{xy} < (\le, c_{xy})$, that
  is, $Z'_{xy} < Z_{xy}$. Now, if $\lleq_{xy}$ is the strict
  inequality $<$, then $\ceil{-Z_{xy}}$ becomes $(<, -c_{xy} + 1)$ and
  hence the condition becomes: $Z'_{xy} + (<, -c_{xy} + 1) < (\le,
  0)$. This is equivalent to saying $Z'_{xy} < (<, c_{xy})$. In both
  cases, the first inequality of Equation (\ref{eq:thm-1}) becomes
  $Z'_{xy} < Z_{xy}$.

  By a similar reasoning, the second inequality of Equation
  (\ref{eq:thm-1}) can be seen to correspond to $Z'_{xy} + (<, -L_y) <
  Z_{x0}$. This proves the theorem.
\end{proof}
 
The $Z \not\incl \alu(Z')$ test involves a comparison of corresponding
edges in the distance graphs $G_Z$ and $G_{Z'}$, so it takes in the
worst case a $\Oo(|X|^2)$ number of steps. Notice that in fact the
test requires only two tests for every pair of clocks.


%% file: fig/zone-dist-graph.tex
\begin{tikzpicture}[scale=0.7]

  \begin{scope}
  \draw[->, very thick, >=stealth] (0,0) -- (0,5); \draw[->, very
  thick, >=stealth] (0,0) -- (6,0);

  \draw (6.3,0) node {\scriptsize $x$} (7, 0); \draw (0, 5.3) node
  {\scriptsize $y$} (0,6.3);

  \draw[very thin, fill=gray, nearly transparent]
  (1,1) -- (1,2) -- (3,4) -- (5,4) -- (5,3) -- (3,1) -- cycle; 
  
  \draw (1,2) -- (4,5);
  \draw[white] (4.5,5.5) -- node[sloped, black] {\footnotesize $y - x
  \lleq_{xy} c_{xy}$} (5.5,6.5) ;
  
\draw (1,0) -- (1,5);
  \draw[white] (1, -0.3) -- node[black] {\footnotesize $0 - x \lleq_{x0} c_{x0}$} (1,
  -0.3);   

  \draw (0,1) -- (7,1);
  \draw[white] (8, 1) -- node[black] {\footnotesize $0-y \lleq_{y0} c_{y0}$} (8.5,1);  

  \draw (3,1) -- (7,5);
  \draw[white] (7.8,5.8) -- node[sloped, black] {\footnotesize $x - y
    \lleq_{yx} c_{yx}$} (8,6);

  \draw (5,0) -- (5,5);
  \draw[white] (4.5, -0.3) -- node[black] {\footnotesize $x - 0 \lleq_{0x} c_{0x}$} (5,
  -0.3); 

  \draw (0,4) -- (7,4);
  \draw[white] (8, 4) -- node[black] {\footnotesize $y -0 \lleq_{0y} c_{0y}$} (8.5,4);

\end{scope}

\begin{scope}[xshift=0.8cm, yshift=2.5cm, scale=1.3]
    \begin{scope}[xshift=7.5cm, yshift=-0.5cm,  scale = 1.2]
      
      \begin{scope}[thick, decoration = {markings, mark=at position
          0.5 with {\arrow{>}}}]
        \draw[postaction={decorate}] (0,0)
          .. controls (0.5,0.5) and (1.5,0.5) .. (2,0);
        \draw[postaction={decorate}] (2,0) .. controls
          (1.5,-0.5) and (0.5,-0.5).. (0,0);
        \draw[postaction={decorate}] (0,0) .. controls
          (1, 1.5) and (3, 1.5) .. (4, 0);

        \draw[postaction={decorate}] (4,0)
          .. controls (3, -1.5) and (1, -1.5) .. (0, 0);
       \draw[postaction={decorate}] (2,0) .. controls
          (2.5, 0.5) and (3.5, 0.5) .. (4, 0);
        
        \draw[postaction={decorate}] (4,0)
          .. controls (3.5, -0.5) and (2.5, -0.5) .. (2, 0);

      \end{scope}

      \foreach \x in {0,1,2} \fill (2*\x, 0) circle
        (2pt); \draw (-0.1,-0.3) node {\scriptsize $0$} (1,-0.3);
        \draw (2,-0.3) node {\scriptsize $x$} (3,-0.4); \draw
        (4.1,-0.3) node {\scriptsize $y$} (5,-0.3); 
      
      \draw (1.2, -0.6) node {\scriptsize $\lleq_{x0}~ c_{x0}$} (2, -0.8); 
      \draw (1.2, 0.6) node {\scriptsize
           $\lleq_{0x}~ c_{0x}$ } (2,0.8);
     \draw (2,1.5) node {\scriptsize $\lleq_{0y} ~c_{0y}$} (3,2); \draw
        (2, -1.5) node {\scriptsize $\lleq_{y0}~ c_{y0}$} (3,-2);
      \draw (2.8, 0.6) node {\scriptsize $\lleq_{xy}~ c_{xy}$} (3, 0.6);
      \draw (2.8, -0.6) node { \scriptsize $\lleq_{yx}~c_{yx}$} (3,
      -0.6);
     
    \end{scope}

\end{scope}
\end{tikzpicture}


%% file: fig/lureg-from-reg.tex
\hspace{0.1cm}
\begin{tikzpicture}[scale = 0.8]

      \draw[very thick] (1,0) -- (1,4);
      \draw (1, -0.4) node {\scriptsize $\mathbf{U_x}$} (1.4, -0.4); 
      \draw[->, very thick, >=stealth] (0,0) -- (0,4); \draw[->, very
      thick, >=stealth] (0,0) -- (5,0);

      \draw (5.3,0) node {\scriptsize $x$} (6, 0); \draw (0, 4.3) node
      {\scriptsize $y$} (0,4.3); \draw (-0.2,-0.2) node {\scriptsize
        $0$} (0,-0.2);

      \begin{scope}[xshift=1cm]
      \fill[gray, nearly transparent] (2,2) -- (2,3) --(1,2) --
      cycle;  
      \draw (2,2) -- (2,3);
      \draw (2,3) -- (1,2);
      \draw (1,2) -- (2,2);
      \draw[dashed] (2,0) -- (2,4);
      \draw[dashed] (0,0) -- (4,4);
      
      \begin{scope}[xshift=-0.7cm, yshift=-0.4cm]
      \draw (2.5, 2.7) node {\scriptsize $v$} (3,2.7);
      \fill (2.3, 2.7) circle (1pt);
      \end{scope}

      \begin{scope}[xshift=-0.2cm, yshift=0.3cm]
      \draw (3.2, 2) node {\scriptsize $\reg{v}$} (4,2); 
      \draw[->, >=stealth] (2.8,2) .. controls (2.6, 2) and (2.4,2)
      .. (2.2,2.2); 
      \end{scope}
      \draw (1.5, -1) node {\footnotesize (a) } (4,-1); 
      \end{scope}

 \end{tikzpicture}
\hfill
\begin{tikzpicture}[scale = 0.8]

      \draw[very thick] (1,0) -- (1,4);
      \draw (1, -0.4) node {\scriptsize $\mathbf{U_x}$} (1.4, -0.4); 
      \draw[->, very thick, >=stealth] (0,0) -- (0,4); \draw[->, very
      thick, >=stealth] (0,0) -- (5,0);

      \draw (5.3,0) node {\scriptsize $x$} (6, 0); \draw (0, 4.3) node
      {\scriptsize $y$} (0,4.3); \draw (-0.2,-0.2) node {\scriptsize
        $0$} (0,-0.2);

      \begin{scope}[xshift=1cm]
      \fill[gray, nearly transparent] (1,2) -- (2,3) -- (4,3) --
      (4,2) -- cycle;  
      \draw (1,2) -- (2,3);
      \draw (2,3) -- (4,3);
      \draw (4,2) -- (1,2);

      \begin{scope}[xshift = 0.2cm, yshift = -0.3cm]
      \draw (2.5, 2.7) node {\scriptsize $v$} (3,2.7);
      \fill (2.3, 2.7) circle (1pt);
       \draw (3.3, 2) node {\scriptsize $\lureg{v}$} (4,2); 
      \draw[->, >=stealth] (2.8,2) .. controls (2.6, 2) and (2.4,2) .. (2.2,2.2); 
      \end{scope}

      \draw (1.5, -1) node {\footnotesize (b)} (4,-1); 

      \end{scope}

 \end{tikzpicture}

\vspace{0.75cm}

\begin{tikzpicture}[scale = 0.8]

\begin{scope}[yshift=-0.8cm]
     
      \draw[thick] (0,1) -- (5,1);
      \draw (-0.4, 1) node {\scriptsize $\mathbf{L_y}$} (0,1); 
      \draw[->, very thick, >=stealth] (0,0) -- (0,4); \draw[->, very
      thick, >=stealth] (0,0) -- (5,0);

      \draw (5.3,0) node {\scriptsize $x$} (6, 0); \draw (0, 4.3) node
      {\scriptsize $y$} (0,4.3); \draw (-0.2,-0.2) node {\scriptsize
        $0$} (0,-0.2);

      
      \fill[gray, nearly transparent] (2,2) -- (2,3) --(3,3) --
      cycle;  
      \draw[very thin] (2,2) -- (2,3);
      \draw[very thin] (2,3) -- (3,3);
      \draw[very thin] (2,2) -- (3,3);
      \draw[dashed] (0,2) -- (5,2);
      \draw[dashed] (0,0) -- (4,4);
     
      \draw (2.5, 2.7) node {\scriptsize $v$} (3,2.7);
      \fill (2.3, 2.7) circle (1pt);
      \draw (3.6, 2.4) node {\scriptsize $\reg{v}$} (4,2); 
      \draw[->, >=stealth] (3.2,2.3) .. controls (3, 2.3) and
      (2.8,2.3) .. (2.6,2.5); 
      \draw (2, -1) node {\footnotesize (c) } (4,-1); 
      \end{scope}

\end{tikzpicture}
\hfill
\begin{tikzpicture}[scale=0.8]
      \draw[very thick] (0,1) -- (5,1);
      \draw (-0.4, 1) node {\scriptsize $\mathbf{L_y}$} (0,1); 
      \draw[->,very thick, >=stealth] (0,0) -- (0,4); \draw[->, very
      thick, >=stealth] (0,0) -- (5,0);

      \draw (5.3,0) node {\scriptsize $x$} (6, 0); \draw (0, 4.3) node
      {\scriptsize $y$} (0,4.3); \draw (-0.2,-0.2) node {\scriptsize
        $0$} (0,-0.2);

      \fill[gray, nearly transparent] (2,1) -- (2,3) -- (3,3) --
      (3,1) -- cycle;
      \draw[very thin] (2,1) -- (2,3);
      \draw[very thin] (2,3) -- (3,3);
      \draw[very thin] (3,3) -- (3,1);

      \draw (2.5, 2.7) node {\scriptsize $v$} (3,2.7);
      \fill (2.3, 2.7) circle (1pt);
      \begin{scope}[xshift = 0.5cm]
       \draw (3.3, 2) node {\scriptsize $\lureg{v}$} (4,2); 
      \draw[->, >=stealth] (2.8,2) .. controls (2.6, 2) and (2.4,2) .. (2.2,2.2); 
      \end{scope}

      \draw (2, -1) node {\footnotesize (d)} (4,-1); 
\end{tikzpicture}


%% file: fig/bounded-regions.tex
\begin{tikzpicture}

    
  \begin{scope}[very thin]
    \foreach \x in {1,...,3}
    {
      \draw (\x,0) -- (\x,4);
      \draw (\x,0) -- (3, 3- \x);
      }
      \draw (0,0) -- (2,2);
      \draw (0,1) -- (1,2);
      \draw (0,1) -- (5,1);
    \end{scope}

      \draw[thick] (3,0) -- (3,4); \draw[thick] (0,2) -- (5,2);

      \draw[->, very thick, >=stealth] (0,0) -- (0,4); \draw[->, very
      thick, >=stealth] (0,0) -- (5,0);

      \draw (5.3,0) node {\scriptsize $x$} (6, 0); \draw (0, 4.3) node
      {\scriptsize $y$} (0,4.3); \draw (-0.2,-0.2) node {\scriptsize
        $0$} (0,-0.2);

      
      \draw[fill=gray!50!white, nearly transparent] 
      (1,0) -- (5,4) -- (5,0) -- cycle; 
      \draw (4, 1.5) node {\footnotesize $Z$} (4.5, 1.5);

      \fill[gray!50!black, nearly transparent] (1,1) -- (2,2) -- (1,2) --
      cycle; 
      \draw (1.3, 1.7) node {\footnotesize $A$} (2,1.2);
      \fill[gray!50!black, nearly transparent] (2,0) -- (3,1) -- (3,0) --
      cycle; 
      \draw (2.6, 0.3) node {\footnotesize $B$} (3,0.2);
\end{tikzpicture}


%% file: fig/min-lureg.tex
\begin{tikzpicture}



      \draw[->, very thick, >=stealth] (0,0) -- (0,6); \draw[->, very
      thick, >=stealth] (0,0) -- (8,0);
      \draw (3, -1.5) node {\footnotesize (a) } (3.5, -1.5); 

      \node at (8.3,-0.3) {\scriptsize $x$};
      \node at (-0.3, 6.3) {\scriptsize $y$}; 
      
      \draw (-0.2,-0.2) node {\scriptsize
       $0$} (0,-0.2);

     \draw (0, 5) -- (8,5);
     \draw (7, 0) -- (7, 6);
     \node [left] at (0,5) {\scriptsize $M_y$};
     \node [below] at (7, 0) {\scriptsize $M_x$};
      
      \begin{scope}[xshift=1cm]
        \fill[gray!50!white, draw=black] (2,2) -- (2,1) -- (4,1) --
        (5,2) --(5,4) -- (4,4) -- cycle;

        \draw[thick] (0,0) -- (4,4);
        \fill (3,3) circle (2pt);
        \node at (3.2, 2.8) {\scriptsize $v_4$};
        \draw[thick] (1,0) -- (5,4); 
        \fill (3.5,2.5) circle (2pt);
        \node at (3.7, 2.3) {\scriptsize $v_3$};
        \draw[thick] (2,0) -- (5,3);
        \fill (4, 2) circle (2pt);
        \node at (4.2, 1.8) {\scriptsize $v_2$};
        \draw[thick] (3,0) -- (5,2);
        \fill (4.5, 1.5) circle (2pt);
        \node at (4.7, 1.3) {\scriptsize $v_1$};
        \fill (3.7, 3.2) circle (2pt);
        \node [above] at (3.9, 3.2) {\scriptsize $v_5$};
        
        \draw[white] (1.8,-1.2) -- node[black,sloped]
        {\scriptsize $\reg{v_1}_{yx} = 4$} (3,0);
        \draw[white] (0.8,-1.2) --
        node[black, sloped] {\scriptsize $\reg{v_2}_{yx} = 3 $} (2,0);
        \draw[white] (-0.2,-1.2) -- node[black, sloped]
        {\scriptsize $\reg{v_3}_{yx} = 2$} (1,0); 
        \draw[white] (-1.2,-1.2) --
        node[black, sloped] {\scriptsize $\reg{v_4}_{yx} = 1$} (0,0);
        \draw[white] (3, 3.5) -- node[black, sloped] {\scriptsize $y
        - x \le -1$} (4, 4.5);
      \end{scope}
\end{tikzpicture}


%% file: conclusions.tex
\section{Conclusions}

Traditional methods for timed automata reachability \emph{store
  abstractions} of zones for termination. Therefore, only convex
abstractions have been used in implementations. We have proposed to
store zones and use abstractions indirectly by means of inclusion
tests $Z \incl \abs(Z')$. This allows us to use non-convex
abstractions while still working with zones. The coarser the
abstraction $\abs$, the higher is the possibility of inclusion, and
hence smaller would be the reachability tree that is explored. For
this construction to work efficiently, one also needs an efficient
inclusion test.

In this paper, we have given a complete solution to using the $\alu$
abstraction for the reachability algorithm. Firstly, we have given an
$\Oo(|X|^2)$ test for $Z \incl \alu(Z')$ inclusion. This is the
cornerstone of our approach since this test is used in the inner
loop of the algorithm. Our test has the same complexity as
$Z \incl Z'$ test used in the traditional algorithm. We have also
shown that the $\alu$ abstraction is the coarsest abstraction that
is sound and complete with respect to reachability for all automata
with the same $LU$-bounds. The result showing
that $\alu$ abstraction is the coarsest possible is quite
unexpected. It works thanks to the observation that when doing forward
exploration it is enough to consider only time-elapsed zones. This
result explains why after $\Extra_{LU}^+$
from~\cite{Behrmann:STTT:2006} there have been no new abstraction
operators~\cite{bouyer-hab2009}. Indeed it is not that easy to find a
better zone inside $\alu$ abstraction than that given by
$\Extra_{LU}^+$ abstraction.

The maximality result for $\alu$ shows that to improve reachability
testing even further we will need to look at new structural properties
of timed automata, or to consider more refined algorithms than forward
exploration. A work in this direction is~\cite{CAV}.
